\documentclass[sigconf]{acmart}

\AtBeginDocument{%
  }

\setcopyright{acmcopyright}
\copyrightyear{2018}
\acmYear{2018}
\acmDOI{XXXXXXX.XXXXXXX}

\acmConference[Conference acronym 'XX]{Make sure to enter the correct
  conference title from your rights confirmation emai}{June 03--05,
  2018}{Woodstock, NY}
\acmPrice{15.00}
\acmISBN{978-1-4503-XXXX-X/18/06}
\copyrightyear{2023}

\usepackage{enumerate}   
\usepackage{wrapfig}    
\usepackage{graphicx}
\usepackage{subfloat}
\usepackage{subfigure}
\usepackage{multirow}
\usepackage{array}
\usepackage{csquotes}
\usepackage{algorithm}
\usepackage{amsmath}
\usepackage{amsthm}
\usepackage{algcompatible} 
\usepackage{caption} 
\usepackage{comment} 
\usepackage[title]{appendix}

\newtheorem{theorem}{Theorem}%  meant for continuous numbers
\newtheorem{proposition}{Proposition}
\newtheorem{lemma}{Lemma}
\theoremstyle{thmstyletwo}%
\newtheorem{definition}{Definition}%

\begin{document}

\title{Load Balanced Demand Distribution under Overload Penalties}

\author{Sarnath Ramnath}
\affiliation{%
  \institution{St. Cloud State University, USA}}  
\email{sarnath@stcloudstate.edu}

\author{Venkata M. V. Gunturi}
\affiliation{%
  \institution{University of Hull, UK} \institution{IIT Ropar, Punjab, India}}
\email{v.gunturi@hull.ac.uk} 

\author{Subi Dangol} %\thanks{^1 Supported\  in\  part\ by\ SCSU\ Husky\ SURF\  award}
\affiliation{%
  \institution{St. Cloud State University, USA}} 
\email{sdangol3@stcloudstate.edu}
 
  \author{Abhishek Mishra}
\affiliation{%
  \institution{IIT Ropar, Punjab, India}}
\email{2018csm1002@iitrpr.ac.in}

 \author{Pradeep Kumar}
\affiliation{%
  \institution{IIT Ropar, Punjab, India}}
\email{2019CSM1008@iitrpr.ac.in}

\renewcommand{\shortauthors}{Ramnath et al.}

\begin{abstract}
Input to the Load Balanced Demand Distribution (LBDD) consists of the following: (a) a set of public service centers (e.g., schools); (b) a set of demand (people) units and; (c) a cost matrix containing the cost of assignment for all demand unit-service center pairs. In addition, each service center is also associated with a notion of capacity and a penalty which is incurred if it gets overloaded. Given the input, the LBDD problem determines a mapping from the set of demand units to the set of service centers. The objective is to determine a mapping that minimizes the sum of the following two terms: (i) the total assignment cost between demand units and their allotted service centers and, (ii) total of penalties incurred. The problem of LBDD finds its application in the domain of urban planning. An instance of the LBDD problem can be reduced to an instance of the min-cost bi-partite matching problem. However, this approach cannot scale up to the real world large problem instances. The current state of the art related to LBDD makes simplifying assumptions such as infinite capacity or total capacity being equal to the total demand. This paper proposes a novel allotment subspace re-adjustment based approach (ASRAL) for the LBDD problem. We analyze ASRAL theoretically and present its asymptotic time complexity. We also evaluate ASRAL experimentally on large problem instances and compare with alternative approaches. Our results indicate that ASRAL is able to scale-up while maintaining significantly better solution quality over the alternative approaches. In addition, we also extend ASRAL to para-ASRAL which uses the GPU and CPU cores to speed-up the execution while maintaining the same solution quality as ASRAL. %In addition, our analysis also reveals that ASRAL can be trivially generalised to return an optimal solution (with time complexity  drastically lower than that of min-cost bi-partite matching) when service centers are not allowed to be overloaded.}
\end{abstract}

\begin{CCSXML}
<ccs2012>
   <concept>
       <concept_id>10002951.10003227.10003236.10003237</concept_id>
       <concept_desc>Information systems~Geographic information systems</concept_desc>
       <concept_significance>500</concept_significance>
       </concept>
 </ccs2012>
\end{CCSXML}

\ccsdesc[500]{Information systems~Geographic information systems}

\keywords{spatial networks, voronoi diagrams, graph algorithms}

\maketitle

\section{Introduction}
\label{intro}
The problem of Load Balanced Demand Distribution (LBDD) takes the following three as input. (a) a set $S$ of service centers (e.g., COVID clinics, schools etc.); (b) a set $D$ of demand units (e.g., people); (c) a cost matrix $\mathcal{CM}$ which contains the cost of assigning a demand unit $d_i\in D$ to a service center $s_j \in S$ (for all $<d_i,s_j>$ pairs). Additionally, each service center $s_j \in S$ is associated with a positive integer capacity and a notion of ``penalty'' which denotes the \textit{"extra cost"} that must be paid to overload the particular service center. Given the input, the LBDD problem determines an mapping between the set of demand units and the set of service centers. The objective here is to determine a mapping which minimizes the sum of the following two terms: (1) total cost of assignment (accumulated across all assignments) and, (2) total penalty incurred (if any) while overloading the service centers.

\noindent \textbf{Problem Motivation:} 
The Load Balanced Demand Distribution (LBDD) problem finds its application in the domain of urban planning. For instance, consider the task of defining the geographic zones of operation of service centers such as schools, COVID clinics and walk-in COVID-19 testing centers (in case of continuous monitoring of the disease at a city scale). 

The key aspect over here being that each of the previously mentioned type of service centers is associated with a general notion of capacity. This capacity dictates the number of people (demand) that can be accommodated (comfortably) each day, week or during any specific duration of time (e.g., typical duration of sickness of patients). In addition, the quality of service at any of these service centers is expected to degrade if significantly more number of people (beyond its capacity) are assigned to it. Now, determining a ``optimal'' mapping between the demand and the service centers under such constraints can be mapped to our LBDD problem. Following is an illustrative instance of our LBDD problem with walk-in COVID-19 testing centers. One can easily define a LBDD instance for schools as service centers in a similar fashion.
 
Consider the task of defining the region of operation for each of the walk-in COVID-19 testing centers in a city. As one can imagine, any testing center would have only  limited resources in terms of testing-kits, lab technicians and other relevant equipment. These infrastructural limits would define its capacity in terms of number of people who can be tested each day. And the overall service-time is likely to get affected if significantly more people (beyond its capacity) show up at this clinic. This aspect is modeled as penalty. Cost of assigning a person to a particular service center could be defined in terms of shortest distance, travel-time or cost of travelling along the shortest path. Now, determining the region of operation under such constraints can be mapped to our LBDD problem. 
 
\noindent \textbf{Challenges:} An instance of LBDD problem can be theoretically reduced to an instance of the min-Cost bipartite matching problem \cite{NET20477}. However, optimal algorithms for min-cost bipartite matching fail to scale up for large LBDD problem instances. For instance, our experiments revealed that for an LBDD instance with just $13$ service centers and around $4000$ demand vertices, the optimal algorithm for min-cost bipartite matching takes around $7$ hours for execution and \emph{occupies about $40$ GB in RAM}. From this, one can easily expect that the optimal algorithm for min-cost bipartite matching algorithm cannot scale-up to large city-scale problem instances. 

\noindent \textbf{Limitations of Related Work}
The current state of the art most relevant to our work includes the work done in the area of network voronoi diagrams without capacities \cite{okabe,Demiryurek2012}, network voronoi diagrams under capacity constraints \cite{Yang2013,7123646,U2010,U2008}, weighted voronoi diagrams \cite{pdiag} and optimal location queries (e.g., \cite{Yao2014,Xiao2011,Diabat16}).  

Work done in the area of network voronoi diagrams without capacities \cite{okabe,Demiryurek2012} assume that the service centers have infinite capacity, an assumption not suitable in many real-world scenarios. On the other hand, work done in the area of network voronoi diagrams with capacities \cite{7123646,U2010,U2008,Yang2013} did not consider the notion of ``overload penalty.'' They perform allotments (of demand units) in an iterative fashion as long as there exists a service center with available capacity. In other words, the allotments stop when all the service centers are full (in terms of their capacity). This paper considers the problem in a more general setting in the sense that we allow the allotments to go beyond the capacities of the service centers. And after a service center is full, we use the concept of the \emph{overload penalties} for guiding the further allotments.

%In \cite{U2010,NET20477}, the problem of capacitated matching on bipartite graphs (``network equivalent'' of Voronoi Diagrams) was explored. Novel optimal and approximation algorithms were proposed in \cite{U2010}. \cite{NET20477} proposed an approximation algorithm for the min-max objective function and formulated an equivalent matching problem for the min average (or min sum) objective function which can solved by known optimal algorithms. However, in both \cite{U2010} and \cite{NET20477}, one would have to create $|V_d|$ copies of the service centers internally for modeling the penalty functions. This is because, the total cost (distance + penalty) of assigning any particular demand node $v$ to a service center $s_i$ depends on the number of demand nodes assigned to $s_i$ before $v$. Consequently, the priority queue sizes in these algorithms attain a huge size of $|V_d|\times|V_s|$ items.   

Weighted voronoi diagrams \cite{pdiag} are specialized voronoi diagrams. In these diagrams, the cost of allotting a demand unit $x$ to a service center $p$ is a linear function of the following two terms: (i) distance between $x$ and $p$ and, (ii) a \emph{real number} denoting the weight of $p$ as $w(p)$. LBDD problem is different from weighted voronoi diagrams. Unlike the weighted voronoi diagrams, our ``$w(p)$'' is a \emph{function} of the number of allotments already made to the service center $p$. And it would return a non-zero value only when the allotments cross beyond the capacity. Whereas in \cite{pdiag}, $w(p)$ is assumed to play its role throughout. 
 
Work done in the area of optimal location queries (e.g., \cite{Yao2014,Diabat16}) focus on determining a new location to start a new facility while optimizing a certain objective function (e.g., total distance between clients and facilities). Whereas, in LBDD, we already have a set of facilities which are up and running, and we want to load balance the demand around them.

\noindent \textbf{Our of Contributions:} This paper makes the following contributions:

\noindent (a) We propose the concept of allotment subspace re-adjustment for the LBDD problem.

\noindent (b) Using the concept of allotment subspace re-adjustment, we propose a novel \emph{ASRAL algorithm} for solving the LBDD problem.

\noindent (c) We evaluate \emph{ASRAL} both theoretically and experimentally (via large problem instances).

\noindent (d) Our experimental results indicate that ASRAL is able to scale-up while maintaining \emph{significantly} better solution quality over the alternative approaches. 

\noindent (e) We also extend ASRAL algorithm to develop a parallel algorithm which can take advantage of increasingly available multi-core and multi-processor systems. Our parallel implementation of ASRAL, called as \emph{para-ASRAL}, takes advantage of both GPU and CPU cores to decrease the execution time while maintaining the same solution quality as ASRAL. 

\noindent (f) Our theoretical analysis shows that ASRAL can be trivially generalized to obtain optimal result when service centers are constrained in terms of capacity, i.e., they are not allowed to be overloaded.  

%\noindent (f) We also prove that our allotment subspace approach can be easily extended to a dynamic case where the allotment needs to be maintained under addition and deletion of demand units. 

\noindent \textbf{Outline:} The rest of this paper is organized as follows: In Section \ref{bc}, we provide the basic concepts and the problem statement. Section \ref{pa} presents our proposed approach. In Section \ref{pasral}, we present our ideas for a parallel implementation (para-ASRAL algorithm) of the ASRAL algorithm. Section \ref{ana} provides detailed theoretical analysis and complexity of our proposed approach. Section \ref{exp} provides an experimental evaluation of our approaches (and the related work). We conclude the paper in Section \ref{con}. Appendix \ref{optlbdd} presents the a trivial generalization of the ASRAL algorithm which can compute the optimal solution when the service centers are not allowed to be overloaded.

\begin{table}[ht]  
\centering
 \begin{tabular}{|m{2cm} | m{5cm} |} 
 \hline 
 \smallskip
Symbol used & Meaning \\
 \hline  
 \smallskip
 $D$ & Set of $n$ demand units given in input \\ \hline
 $S$ & Set of $k$ service centers given in input \\ \hline 
 %$c_{s_i}$ & Capacity of a service center $s_i$  \\ \hline
 $q_{s_i}()$ & Penalty function of a service center $s_i$ \\ \hline
 $\mathcal{A},\mathcal{B},\mathcal{F}$ & We use mathcal letters to denote an allotment set. An allotment contains the pairs $<d_i,s_j>$. Here, demand unit $d_i$ is allotted to service center $s_j$ \\ \hline
 $\mathcal{CM}(d_i,s_j)$ & Cost (positive integer) of assigning demand unit $d_i$ to  service center $s_j$ \\ \hline
 % $\Delta(\mathcal{A})$ & Objective function value (Equation \ref{eq2}) of the allotment $\mathcal{A}$ \\ \hline
 $\Gamma(\mathcal{A})$ & Allotment subspace graph of the allotment $\mathcal{A}$. Nodes of this graph are service centers and edges represent transfer of demand units across service centers. \\ \hline
 %$NCyc\Gamma(\mathcal{A})$  & \\ \hline
 %$NPath\Gamma(\mathcal{A})$  & \\ 
  \hline
 \end{tabular}
 \caption{Notations used in the paper}
 \label{notations}
\end{table}

\section{Basic Concepts and Problem Definition}
\label{bc}
\begin{definition}
\textbf{A Service Center} is a public service unit of a particular kind (e.g., schools, hospitals, COVID-19 testing centers in a city). A set of service centers is represented as $S=\{s_1,..., s_{k}\}$.  
\end{definition}

\begin{definition}
\textbf{A Demand unit} represents a unit population which is interested in accessing the previously defined service center. A set of demand units is represented as  $D=\{{d_1},..., d_{n}\}$, $n$ is the number of demand units. %For simplicity, we assume that all population of a city resides on \textbf{vertices} of the graph $G$. One can easily generalize to case where the population is staying on the edges by creating dummy nodes (of degree two) at all places were there is some amount of population on the edge. 
\end{definition}

\begin{definition}
\textbf{Capacity of a service center} $(c_{s_i})$ is the prescribed number of  demand units that a service center $s_i$ can accommodate comfortably. For example, the number of people a COVID-19 testing facility can test in a day (or a week) would define its capacity.  Similarly, the number of students a school can admit would correspond to the notion of capacity defined in this paper.
\end{definition}

\begin{definition}
\label{pdef}
\textbf{Penalty function for a service center} $(q_{s_i}())$ is a function which returns the ``extra cost'' ($>0$) that must be paid for every new allotment to the service center $s_i$ which has already exhausted its capacity $c_{s_i}$. If allotment is done within the capacity of a service center, then no penalty needs to be paid. Examples of penalty cost in real world could be cost to add additional infrastructure and/or faculty in a school. 

Penalty function takes into account the current status of $s_i$ (i.e., how many nodes have been already added to $s_i$) and then returns a penalty for the $j^{th}$ ($1\leq j \leq (n-c_{s_i})$) allotment beyond its capacity. $q_{s_i}()$ returns only positive non-zero values and is monotonically increasing over $j$ ($1\leq j \leq (n-c_{s_i})$). The intuition is that one may need to add increasingly more resources to a school (or a clinic) as the overloading keeps increasing.
\end{definition}
 
\begin{definition} 
\textbf{Demand-Service Cost Matrix $\mathcal{CM}$:} contains the cost (a positive integer) of assigning a demand unit $d_i\in D$ to a service center $s_j \in S$ (for all $<d_i,s_j>$ pairs). The assignment cost can represent things such as shortest distance over the road network, travel-time and/or cost of traveling from $d_i$ to $s_j$ along the shortest path (in Dollars), Geodetic distance or Euclidean distance, etc. 
\end{definition}

\begin{definition}
\textbf{Allotment:} An allotment is a set containing pairs $<d_i,s_j>$. Here, demand unit $d_i$ is allotted to service center $s_j$. We use mathcal letters (e.g., $\mathcal{A},\mathcal{B}$) to denote an allotment.
\end{definition}

\begin{definition}
\textbf{Complete Allotment:} An allotment is said to be an complete allotment iff it has cardinality $n$ (i.e., number of demand nodes). 
\end{definition}

\subsection{Problem Statement}
\label{probdef}
We now formally define the problem of load balanced demand distribution by detailing the input, output and the objective function:

\noindent \textbf{Given}:			
			\begin{itemize}
			\item A set of service centers $S=\{s_1,..., s_{k}\}$.
			\item A set of demand units  $D=\{{d_1},..., d_{n}\}$.
			\item A demand-service cost matrix $\mathcal{CM}$ which contains the cost of allotting a demand unit $d_i \in D$ to a service center $s_j \in S$ ($\forall d_i \in D,~\forall s_j \in S$). 
			\item Capacity $(c_{s_i})$ of each service center $s_i \in S$.  
			\item Penalty function $(q_{s_i}())$ of each service center $s_i \in S$.
			\end{itemize}
			
\noindent \textbf{Output}: A complete allotment which minimizes the following objective function. 
			
\noindent \textbf{Objective Function}: 
			\begin{multline}
			 \textit{\textbf{Minimize}}
			 \Bigg\{
			 \sum_{\substack{s_i \in Service \\Centers}} 	
				\Bigg\{
				\sum_{\substack{d_j \in Demand  \ unit\\ allotted \ to \  s_i}} 
					\mathcal{CM}(d_j,s_i)
					\Bigg\} \\ + Total~Penalty~across~all~s_i \Bigg\}
			\label{eq2}   
			\end{multline}

\subsection{Instantiating LBDD problem in real life}		
For instantiating LBDD problem real life, we need to define our assignment cost (in $\mathcal{CM}$) and penalties. As mentioned earlier, assignment cost can be defined in terms such as shortest distance, cost of traveling along the shortest path, geodetic distance, etc. The penalty values can be defined along the terms of cost (in dollars) of adding more infrastructure to the service center or additional delay in response caused due to overload. Note that for maintaining easy interpretability of results, we should have common  ``units of measurement'' for both the assignment cost and penalty. Cost (in Dollars) could be one such common unit of measurement when we are working with public service units such as COVID clinics, COVID-19 testing centers, schools, etc. Appendix \ref{pracLBDD} provides another sample instantiation of the LBDD problem with schools as service centers.

\begin{figure*}[ht] 
\begin{center} 
\subfigure[An arbitrary  allotment.]{\label{ncyca}\includegraphics[width=68mm]{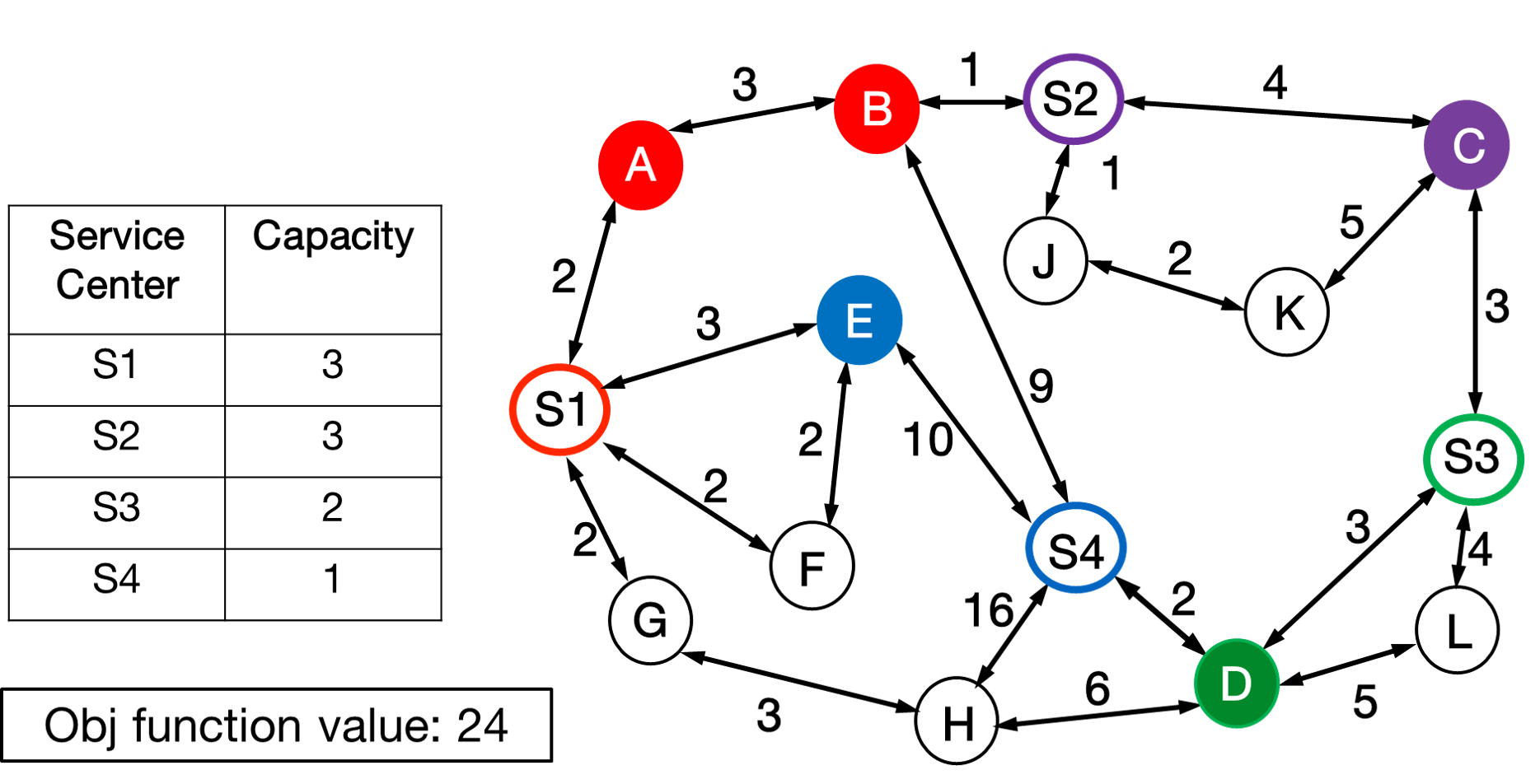}} \hspace{8mm}
\subfigure[Sample edges in the allotment sub-space multi-graph]{\label{ncycb}\includegraphics[width=45mm]{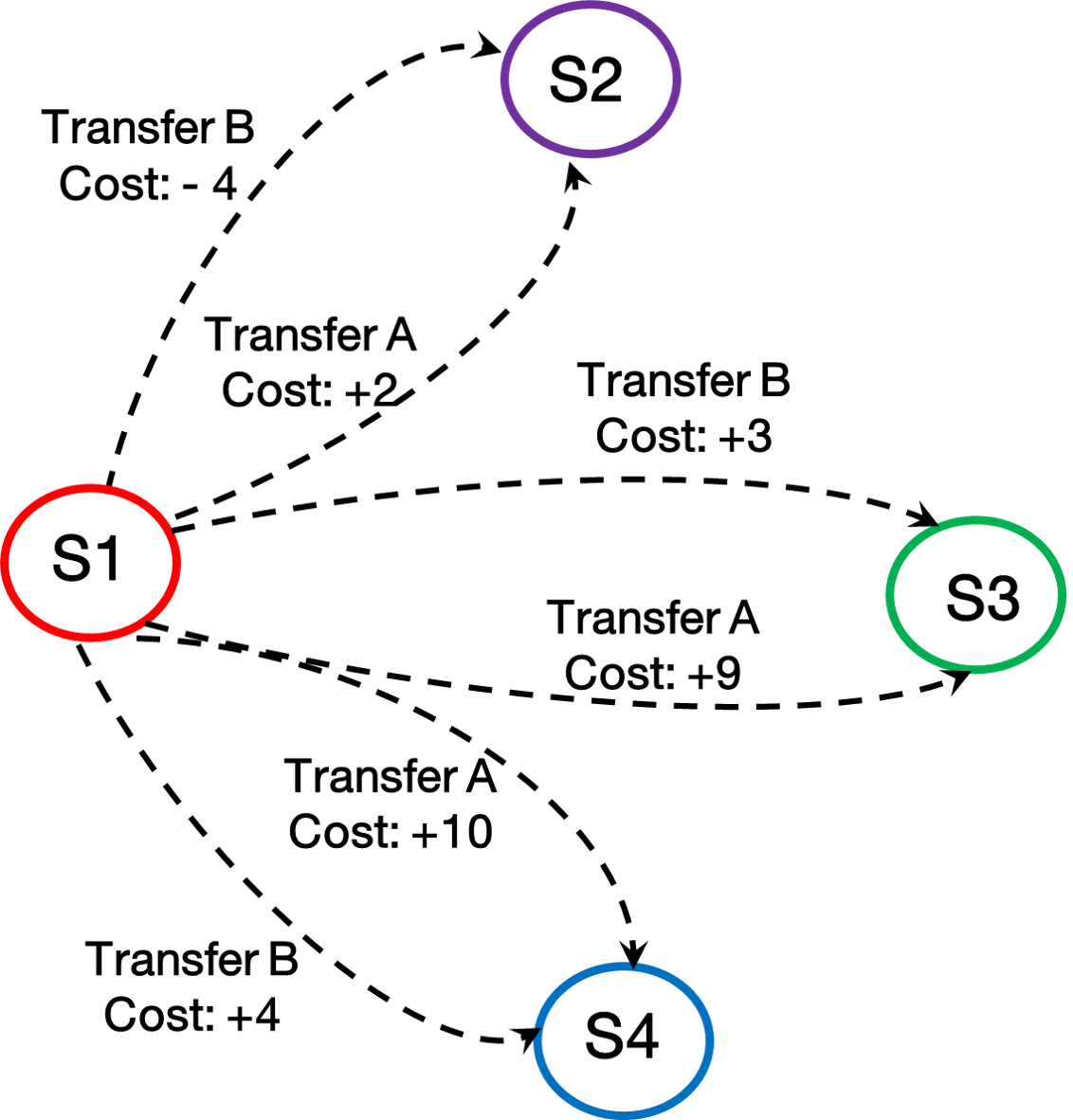}} \hspace{8mm}
\subfigure[A negative cycle]{\label{ncycc}\includegraphics[width=45mm]{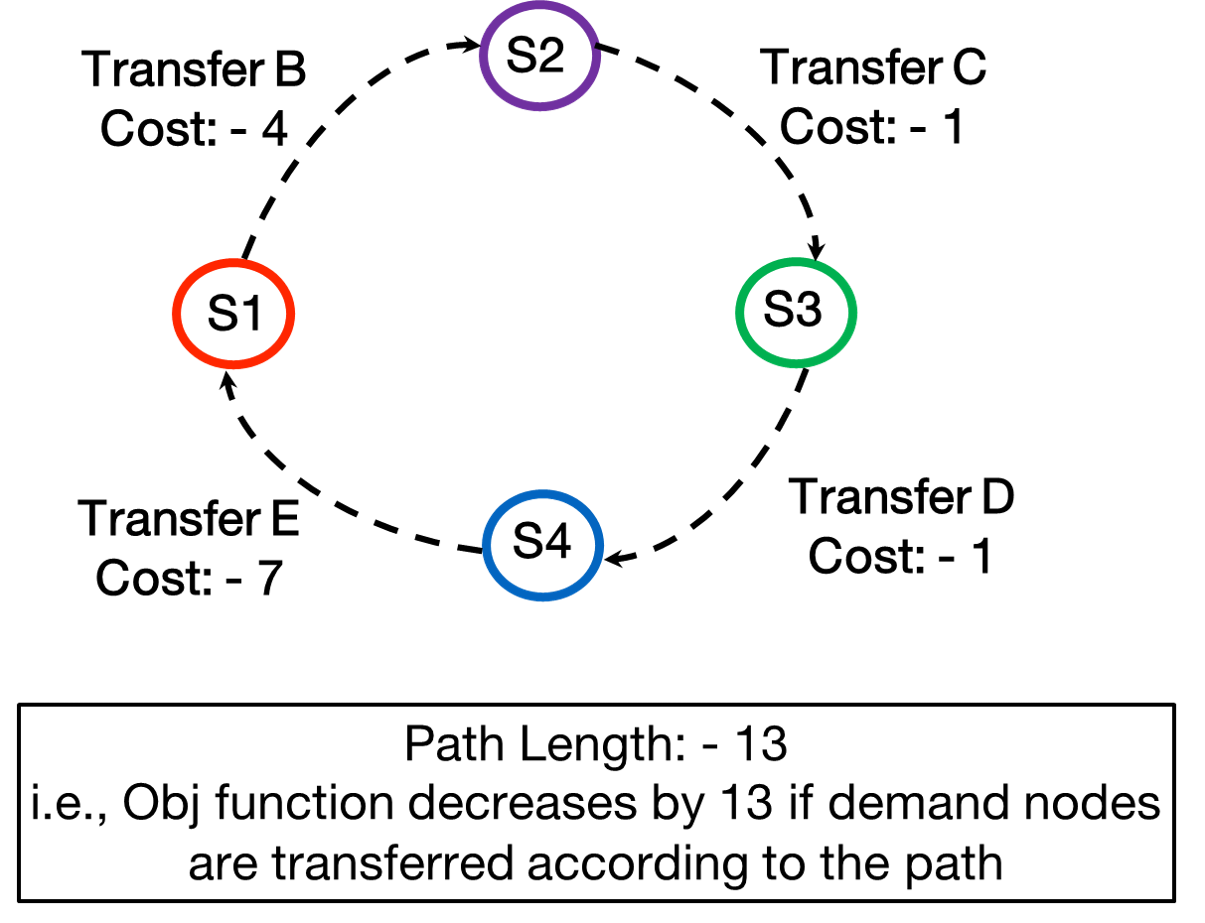}}
\end{center}
\caption{Illustrating allotment subspace multi-graph on sample allotment. Vertices S1, S2, S3 and S4 are service centers. demand units allotted to a service center are shaded using the same color as that of the service center, e.g., demand unit A is allotted to S1 in the sub-figure (a). }
\label{ncyc}  
\vspace{-2mm}   
\end{figure*}

\section{Proposed Approach}
\label{pa}
This section presents our proposed Allotment Subspace Re-adjustment based approach (ASRAL) for the LBDD problem. The remaining of section is organised as follows. Section \ref{over} presents an overview of our proposed approach. In Section \ref{asr}, we introduce our key computational idea of allotment subspace re-adjustment. In Section \ref{ncalgo} and Section \ref{npalgo}, we describe some crucial sub-routines which would be used later in the our ASRAL algorithm. In Section \ref{algodetails}, we present the details of our proposed ASRAL algorithm. We present a parallel implementation of ASRAL (called as para-ASRAL) in Section \ref{pasral}.

\subsection{Overview of the algorithm}
\label{over}
Overall, our algorithm follows an incremental strategy to build the solution. Initially, we determine the ``most suitable'' service center (using $\mathcal{CM}$ given in the input) for all the demand units $d_i \in D$. The ``most suitable'' for a demand unit $d_i$ is defined as the service center $s_{min}$ with whom $d_i$ has the lowest assignment cost in $\mathcal{CM}$. Following this, we insert the tuples $<$demand unit $d_i$, its corresponding $s_{min}>$ into a heap. This heap is ordered on the assignment cost (to their respective $s_{min}$) of the tuples. In each iteration, the demand unit to be processed next is picked from the heap using the extract-min operation. Let the result of the extract-min operation (in an arbitrary iteration) be the demand unit $d_i$ and its ``most suitable'' service center be $s_{min}$. Now, we have two cases:
 
\noindent $\mathbf{s_{min}}$ \textbf{has free capacity (trivial case):} $d_i$ is assigned to $s_{min}$ and the objective function is increased by a value ``$\mathcal{CM}(d_i,s_{min})$.''

\noindent $\mathbf{s_{min}}$ \textbf{was already full:} Here as well we assign $d_i$ to $s_{min}$. However, this allotment would lead to penalty. Thus, the objective function is increased by a value ``$\mathcal{CM}(d_i,s_{min})$ + penalty at $s_{min}$''

Now, we attempt to adjust our current allotment with the aim of decreasing the objective function value and possibly offset the effect the penalty. These adjustments are of following two types. First, we undertake adjustments which do not lead to any change in total penalty. This is done by removing any negative cycles, in what we refer to as the \emph{allotment subspace multigraph} (details in next section). Following this, we attempt to ``move'' the overload created on $s_{min}$ to a different service center via another sequence of adjustments. This is done by removing the most negative path (originating on $s_{min}$) from the \emph{allotment subspace multigraph}. Note that in order to determine the most negative path, the \emph{allotment subspace multigraph} must not have any negative cycles. Our algorithm guarantees this aspect both after the first set of re-adjustments (i,e., removing negative cycle) and after moving the overload to the best possible location (i.e., removing negative path).

\subsection{Allotment Subspace Re-adjustment} 
\label{asr}
The key aspect of this paper is the concept of re-adjustments in the space of, what we refer to as, the \emph{allotment subspace multi-graph}. Each service center (given in the problem instance) forms a node in this allotment subspace multi-graph. Each edge in the multigraph represents the potential transfer of a demand unit from one service center another. Formally, we define this as follows. 

\begin{definition}
Given an allotment $\mathcal{A}$ comprising of tuples of the form $<d_i,s_j>$ (demand unit $d_i$ is allotted to $s_j$), the allotment subspace multi-graph $\Gamma(\mathcal{A}) = (M,N)$ is defined as follows:
\begin{itemize}
    \item \textbf{Set of nodes} ($M$): Each service center $s_i \in S$ is a node in $M$.
    \item  \textbf{Set of edges} ($N$): Each edge is defined as a triple $e=(s_i,s_j,dn)$. Here, $e$ represents the potential transfer of the demand unit $dn$ (which is currently allotted to $s_i$ in $\mathcal{A}$) from $s_i$ to $s_j$. Cost of $e$ is defined as $\mathcal{CM}(dn,s_j)-\mathcal{CM}(dn,s_i)$. This cost is referred to as the \emph{transfer-cost} of the demand unit $dn$.
\end{itemize}
\label{gammadef} 
\end{definition}

Consider Figure \ref{ncyca} which illustrates an arbitrary allotment for a sample problem instance. Note that Figure just a illustrates an arbitrary allotment (rather than following a procedure) to help in easy illustration of the negative cycle concept. For ease of understanding, the figure illustrates the LBDD instance where both service centers and demand units are present in a road network represented as directed graph. Here, nodes $S1$, $S2$, $S3$ and $S4$ are service centers, whereas other nodes (e.g., $A$, $B$, $C$, etc) are demand units (people). In this problem, cost of assigning a demand unit $v$ to a service center $r$ is the shortest distance between $v$ and $r$. For instance, cost of assigning demand unit $B$ to $S1$ is the shortest distance between $B$ and $S1$ in the graph (which is $5$ in this case). The figure also details the total capacity of each of the service centers. In the assignment shown in Figure \ref{ncyca}, demand units which are allotted a service center are filled using the same color as that of their allotted service center. For e.g., demand units $A$, $B$ are assigned to resource unit $S1$. Nodes which are not yet allotted are shown without any filling.

The allotment subspace graph of this allotment would contain four nodes (one for each service center). And, between any two service centers $s_i$ and $s_j$, it would contain directed edges representing transfer of demand units across the service centers. Edges directed towards $s_j$ (from $s_i$) would represent transfer of demand units to $s_j$ (from $s_i$). 

Figure \ref{ncycb} illustrates some edges of the allotment subspace graph ($\Gamma()$) of the allotment shown in Figure \ref{ncyca}. To maintain clarity, we do not show all the edges in the Figure. Figure \ref{ncycb} illustrates only the edges whose tail node is service center S1. Now, given that S1 was allotted two demand nodes ($A$ an $B$ according to Figure \ref{ncyca}), the node corresponding to S1 in the allotment subspace graph would have six outgoing edges (two edges to each of the other service centers). For instance, consider the two edges directed from S1 to S2 in Figure \ref{ncycb}. One of them represents transfer of demand units $A$ (to S2), and the other one represents the transfer of demand unit $B$ to the service center S2. Cost of the edge is defined as the difference in the assignment cost to the service centers. For instance, cost of edge corresponding to $A$ in Figure \ref{ncycb} is defined as $\mathcal{CM}(A,S2)-\mathcal{CM}(A,S1)$ (which is $2$ in our example). Note that edge costs in the allotment subspace graph may be negative in some cases.

\subsubsection{Using allotment subspace  multigraph for improving the solution}
The allotment subspace  multigraph helps us to improve the current solution in following two ways: (a) Negative cycle removal and, (b) Negative path removal. It is important to note that these two are not completely independent of each other. Before removing negative path, we must first ensure that there are no negative cycles in the allotment subspace multigraph (our algorithm ensures this). We would now discuss these ideas at conceptual level. Details on their implementation are given in Section \ref{ncalgo} and Section \ref{npalgo}.

\noindent \textbf{Adjustment via Negative Cycle Removal:} Consider the cycle shown in Figure \ref{ncycc}. This directed cycle appears in the allotment subspace multigraph of the allotment shown in Figure \ref{ncyca}. Total cost of this cycle happens to be $-13$. In other words, we were to we adjust our allotment (shown in Figure \ref{ncyca}) by transferring the demand units corresponding to the edges in this cycle (e.g., transfer $B$ to S2, transfer $C$ to S3, transfer $D$ to S4, transfer $E$ to S1) then, the objective function value would decrease by $13$ units.  

\begin{figure}[h]  
\begin{center} 
\subfigure[An arbitrary allotment.]{\label{npa}\includegraphics[width=68mm]{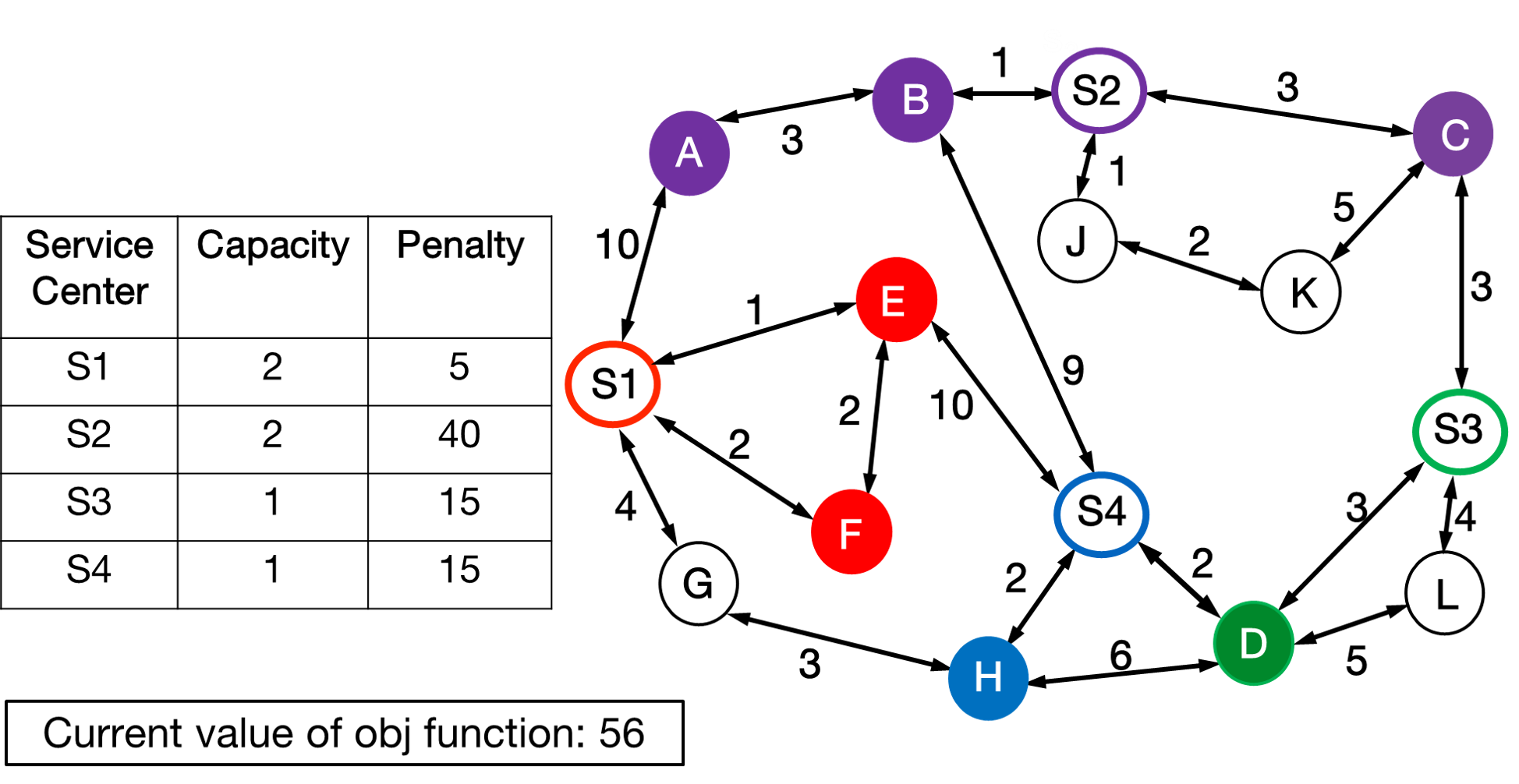}}
\subfigure[A negative path in the allotment sub-space graph]{\label{npb}\includegraphics[width=55mm]{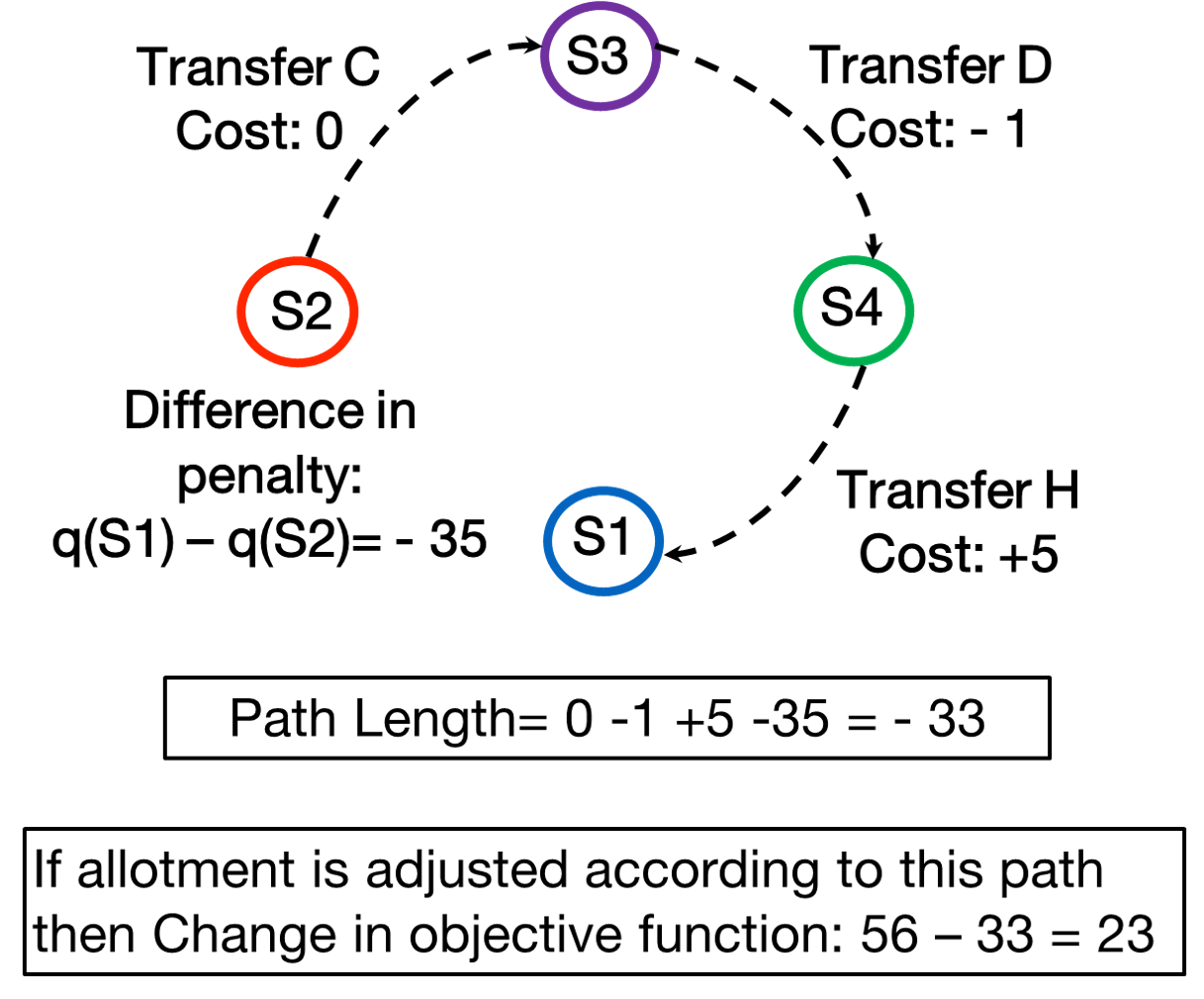}}
\end{center}
\caption{Illustrating negative paths in allotment sub-space multi-graph. Vertices S1, S2, S3 and S4 are service centers. demand nodes allotted to a service center are shaded using the same color as that of the service center.}
\vspace{-2mm}   
\end{figure} 

\noindent \textbf{Adjustment via Negative Path Removal:}  Consider the arbitrary allotment shown in Figure \ref{npa}. Similar to previous case, we used an arbitrary allotment (rather than following a procedure) to help in easy illustration of the negative path concept. Moreover, for sake for easy understanding, problem instance shown in Figure \ref{npa} assumes that penalty functions are just integers instead of monotonic functions as defined in Definition \ref{pdef}.

Figure \ref{npb} illustrates a path in the corresponding allotment subspace multigraph of the allotment shown in Figure \ref{npa}. On this path, we transfer demand unit $C$ to S3, $D$ to S4 and finally, transfer $H$ to S1.  Also note that S2 was initially overloaded (in Figure \ref{npa}). Thus, if we are to adjust the allotment according the mentioned path, then the total penalty paid at S2 would decrease since it is transferring a demand unit ($C$) to another service center. On the other hand, we now have to pay the penalty at S1 as its accepting a demand unit beyond its capacity. Following is an expression of the total cost of this path:

\noindent Cost of Path = TransferCost(C~to~S3) + TransferCost(D~to~S4) + 

\hspace{1.5cm} TransferCost(H~to~S1) + Penaltyat(S1) - Penaltyat(S2)

\hspace{1.5cm} = \textbf{-33}      

Thus, if we adjust our initial allotment (shown in Figure \ref{npa}) according to the path shown in Figure \ref{npb}, then objective function value would decrease by 33. Note that, in our implementation, we do not re-create the $\Gamma()$ after each iteration. We rather maintain it by making incremental changes to the multi-graph. More details are provided later.

 \subsection{Procedure for Removing Negative Cycles}
 \label{ncalgo}
The allotment subspace multigraph ($\Gamma(\mathcal{A})$) of an arbitrary allotment $\mathcal{A}$ can have multiple negative cycles. Removing all of them would be computationally intensive. To this end, in our actual implementation of this concept, we adopt an incremental approach wherein we add demand units to the allotment one by one. And after each new allotment, we determine and remove the negative cycles which could have been created due to this addition.  

More specifically, let $\mathcal{A}$ be the current allotment, and an unprocessed demand unit $d_i$ is now added to a service center $s_j$. Let the resulting allotment be $\mathcal{A'}$. Since the only modification to $\Gamma(\mathcal{A})$ is on the outgoing edges from $s_j$, if $\Gamma(\mathcal{A})$ did not have negative cycles, then any negative cycle in $\Gamma(\mathcal{A'})$ must pass through $s_j$. In Lemma \ref{mainlemma}, we show that if we remove the most negative cycle from $\Gamma(\mathcal{A'})$, then the resulting allotment $\mathcal{A''}$ would be such that we would not have have negative cycles in $\Gamma(\mathcal{A''})$. To compute the most the negative cycle, we identify the ``distinguished'' service center ($s_j$ in our example) and create a variant of the allotment subspace multigraph, called  the \emph{NegCycle} allotment subspace graph. These concepts are formalized in Definition \ref{mincycgamma}.

\begin{definition}
Let an unprocessed demand unit $d_i$ be added to a service center $s*$. Let the resulting allotment be $\mathcal{A}$ and its allotment subspace multigraph be $\Gamma(\mathcal{A})$. The NegCycle allotment subspace graph, $NCyc(\Gamma(\mathcal{A}), s*) = (P,Q)$, is defined as follows:\\ 
\textbf{ (1) Set of nodes} ($P$): Each service center $s_l \in S-{s*}$ is a node in $P$. Service center $s*$ (referred to as the {\em anchor} or {\em distinguished} service center) is represented by two nodes $s*^{in}$ and $s*^{out}$. \\
\textbf{ (2) Set of edges} ($Q$): Cost of the edges between any two nodes $u$ , $v$ $\in P$ is defined as follows:  \\
        (a) if $u \notin \{s*^{in},s*^{out}\}$ and $v \notin \{s*^{in},s*^{out}\}$, then $cost(u,v)$ is the minimum of the costs of all the edges which are directed from $u$ to $v$ in $\Gamma(\mathcal{A})$ (ties are broken arbitrarily). \\
        (b) if $u = s*^{out}$, then $cost(u,v)$ is the minimum of the costs of all the edges which are directed from $s*$ to $v$ in $\Gamma(\mathcal{A})$ (ties are broken arbitrarily). \\
        (c) if $v = s*^{in}$, then $cost(u,v)$ is the minimum of the costs of all the edges which are directed from $u$ to $s*$ in $\Gamma(\mathcal{A})$ (ties are broken arbitrarily).\\
        (d) all other edges have $0$ cost. 
\label{mincycgamma} 
\end{definition}

After the creating of the NegCycle allotment subspace graph ($NCyc\Gamma()$), we determine the lowest cost path between the "in" and "out" copies ($s_j^{in}$ and $s_j^{out}$) of the anchor service center. Note that since both $s_j^{in}$ and $s_j^{out}$ map to $s_j$ in $\Gamma(\mathcal{A})$, this path gives us the cycle of least cost passing through $s_j$ in $\Gamma(\mathcal{A})$. For computing the lowest cost path, we use Bellman Ford's label correcting algorithm. Adjusting the current allotment $\mathcal{A}$ according to the edges of this lowest cost path removes the negative cycle. Note that we cannot use a label setting approach (e.g., Dijkstra's) for finding lowest cost path in $NCyc\Gamma()$ as it may contain negative edges. A detailed pseudo-code of the procedure is given in Algorithm \ref{algrnegc}. 

%\noindent \textsl{Impact of Boundary vertices on $NCyc\Gamma()$:} As mentioned earlier, we consider only boundary vertices if our LBDD instance is coming from a road network. This choice does not have any undesirable effect on $NCyc\Gamma()$. This is because demand units which are inside the allotment region (i.e., not boundary vertices)  would have a higher value of transfer cost than the boundary vertices. Therefore, even if we consider such demand units for edges in $\Gamma(\mathcal{A})$, they would not make it to $NCyc\Gamma()$ as it takes only the edge which has the lowest transfer cost. 

%\noindent \textsl{Impact of $BestTransHeap$ on $NCyc\Gamma()$:} Trivially, $BestTransHeap$ would not have negative any impact on $NCyc\Gamma()$ as well. This is because, in case of non-road network LBDD instance, we propose to maintain only best transfer edge in $\Gamma()$. 

\begin{algorithm}[h]
\caption{Negative Cycle Refinement}
\label{algrnegc}
\begin{flushleft}
\noindent \textbf{Input:} (a) Current allotment $\mathcal{A}$; (b) Anchor service center $SC_\theta$ \\
\noindent \textbf{Output} Resulting allotment $\mathcal{B}$ after adjustment in such that there are no negative cycles in its allotment subspace.
\end{flushleft}
\begin{algorithmic}[1]
\STATE Create the NegCycle allotment subspace graph $NCyc\Gamma(\mathcal{A},SC_\theta)$ with $SC_\theta$ as the anchor vertex.
\STATE Let $w_{out}$ and $w_{in}$ be ``IN'' and ``OUT'' copies of the anchor vertex $SC_\theta$ in $minCyc\Gamma(\mathcal{A})$.
\STATE Find lowest cost path $P_{min}$ in $NCyc\Gamma(\mathcal{A},SC_\theta)$ from $w_{out}$ to $w_{in}$
\IF{Cost of $P_{min} < 0$}
\FOR {each demand node transfer $<s_{old},dn,s_{new}>$ $\in P_{min}$} 
\STATE $\mathcal{B} \leftarrow$ change $\mathcal{A}$ by removing $dn$ from $s_p$ and adding it to $s_q$. \State Update $\Gamma DS$ 
\STATE $\mathcal{A} \leftarrow \mathcal{B}$
\STATE $\Delta \leftarrow$ $\Delta~+$  Cost of corresponding edge in $P_{min}$ /* $\mathcal{CM}(dn,s_p) - \mathcal{CM}(dn,s_q)$ */
\ENDFOR
\ENDIF
\STATE Return new allotment $\mathcal{B}$
\end{algorithmic}
\end{algorithm}

\subsection{Procedure for Removing Negative Path} 
\label{npalgo}
Consider a case where an unprocessed demand unit $d_i$ is added to a service center $s_j$ in some iteration of the main allotment algorithm, and this assignment leads to penalty at $s_j$. Let the resulting allotment be $\mathcal{B}$. As explained earlier, the corresponding $\Gamma(\mathcal{B})$ may have a negative cycle passing through $s_j$. We first remove the negative cycle from $\Gamma(\mathcal{B})$ using the procedure described in the previous section. Let the resulting allotment be $\mathcal{B'}$. Now, we attempt to decrease the value of the objective function by looking for the most negative path which starts at $s_j$ from $\Gamma(\mathcal{B'})$. This is done via the \emph{NegPath allotment subspace graph} defined next.

\begin{definition}
Let an unprocessed demand node $d_i$ be added to a service center $s*$. Let $\mathcal{B}3$ be the resulting allotment after removing the negative cycle from $\Gamma(\mathcal{B})$. The NegPath allotment subspace graph $NPath\Gamma(\mathcal{B}3,s*) = (R,S)$ is defined as follows:\\
    (1) \textbf{Set of nodes} ($R$): Each service center $s_i \in S-{s*}$ is a node in $P$. Service center $s*$ (referred to as the anchor service center) is represented by two nodes $s*_{in}$ and $s*_{out}$. \\
    (2)  \textbf{Set of edges} ($S$): Cost of the edges between any two nodes $u$ , $v$ $\in Q$ is defined as follows: \\
        (a) if $u \notin \{s*_{in},s*_{out}\}$ and $v \notin \{s*_{in},s*_{out}\}$, then $cost(u,v)$ is the minimum of the costs of all the edges which are directed from $u$ to $v$ in $\Gamma(\mathcal{B}')$ (ties are broken arbitrarily).\\
        (b) if $u = s*_{out}$, then $cost(u,v)$ is the minimum of the costs of all the edges which are directed from $s*$ to $v$ in $\Gamma(\mathcal{B}3)$ (ties are broken arbitrarily). \\
        (c) if $v = s*_{in}$, then $cost(u,v)$ = penalty at $u$ ($q_u()$) -  penalty at $s*$ ($q_{s*}()$). \\
         (d) if $u$ and $v$ are same then $cost(u,v)=0$, i.e., self loops have $0$ cost.\\
\label{minpathgamma} 
\end{definition}

After the creating of the NegPath allotment subspace graph ($NPath\Gamma()$), we determine the lowest cost path between the ``in'' and ``out'' copies ($s*_{in}$ and $s*_{out}$) of the anchor service center. For computing the lowest cost path, we use Bellman Ford's label correcting algorithm as $NPath\Gamma()$ can contain edges with negative weights. Adjusting the current allotment (say $\mathcal{B}3$) according to the edges of this lowest cost path removes this negative path. Once the most negative path is removed, we are ready to add the next demand unit. However, we need to be assured that removing the negative path did not introduce any negative cycles. In Lemma \ref{mainlemma2}, we show that if we remove the most negative path from $\Gamma(\mathcal{B}3)$, then the resulting allotment, $\mathcal{B}4$, would be such that we would not have have negative cycles in $\Gamma(\mathcal{B}4.)$ A detailed pseudo-code of the procedure is given in Algorithm \ref{algrnegp} (Appendix \ref{negpalgo}). It is very similar to the pseudo-code of negative cycle refinement.

\subsection{Details of the ASRAL algorithm}
\label{algodetails} 
This section details the operationalization of our previously discussed ideas of allotment subspace re-adjustments via negative cycle and negative paths. We use following data structures in our algorithm.\\
 
\noindent \textbf{Demand-service cost matrix} $\mathcal{CM}$ \textbf{construction}: Depending on the nature of the underlying problem space, one can determine the $\mathcal{CM}$ in following ways. In case both our service centers and demand units are embedded in road networks, one can determine $\mathcal{CM}$ by using Floyd-Warshall algorithm \cite{Cormen:2001:IA:580470}. One can also use other shortest path algorithms (e.g., \cite{phast,goldberg2007better,geisberger2008contraction,bast2007fast,jing1998hierarchical}) for this purpose. On the other hand, if the service center and demand units are spread over continuous geographic space, then $\mathcal{CM}$ can be computed using euclidean distance or geodetic distance.  \\

%\noindent \textbf{Boundary Vertex Maintenance (BVHashMap):} Boundary vertices would be stored in a hash data-structure (referred to as the BVHashMap). The service center id forms the ``key'' in this hash, and a list of boundary vertices forms the ``value'' of this hash data structure. BVHashMap needs to updated whenever a new allotments are made or the allotment is refined using negative cycles and negative paths. Algorithm \ref{algupdateBV} (Appendix \ref{ubvalgo}) details the pseudo-code for this procedure.

\noindent \textbf{Allotment subspace multi-graph implementation ($\Gamma DS$)} In our implementation, $\Gamma(\mathcal{A})$ is maintained using using $2\times$ ${k}\choose {2}$ number of minheaps, where $k$ is the number of service centers. Basically, one heap for each ordered pair of service centers. This heap, \emph{referred to as BestTransHeap}, is ordered on the transfer cost of demand units (across the pair of service centers). For any order pair of service centers <$s_i,s_j$>, the top of its corresponding $BestTransHeap_{ij}$ would contain the demand unit (which is currently allotted to $s_i$) which has the lowest transfer cost to $s_j$. During the course of the algorithm whenever a demand unit is added or deleted from a service center $s_i$, then $k-1$ number of heaps corresponding to $s_i$ are updated accordingly. 

Both $NCyc(\mathcal{A},s*)$ and $NPath\Gamma(\mathcal{A},s*)$ can be created from $\Gamma(\mathcal{A})$. In case of $NCyc\Gamma(\mathcal{A},s*)$, edge between service centers $s_i$ and $s_j$ is essentially, the top of $BestTransHeap_{ij}$ heap. Edges in $NPath\Gamma(\mathcal{A},s*)$ are also constructed in a similar fashion with an exception for the penalty difference edges (refer Definition \ref{minpathgamma}) which can be added in $O(1)$ time. In our implementation, both $NCyc(\mathcal{A},s*)$ and $NPath\Gamma(\mathcal{A},s*)$ are created once, and are maintained thereafter with each demand unit addition, negative cycle and negative path refinements. \\

%We use boundary vertices (in case of road network based LBDD instance) and the $BestTransHeap$s (in case of non-road network space) to keep $\Gamma DS$ updated.  

Algorithm \ref{mainalgo} presents a pseudo-code of our proposed \textit{A}llotment \textit{S}ubspace \textit{R}e-adjustment based \textit{A}pproach for \textit{L}BDD  (\emph{ASRAL}) algorithm. We start our algorithm by determining the ``most suitable'' service center (using $\mathcal{CM}$) for each of the demand unit $d_i$. The ``most suitable'' service center $s_{min}$ for a demand unit $d_i$ is defined as the service center which has the lowest assignment cost for $d_i$. Following this, we insert the tuples $<$demand unit $d_i$, its corresponding $s_{min}>$ into a heap (referred to as the \textit{MinDistance heap}). This heap is ordered on the assignment cost (to their respective $s_{min}$) of the tuples. In each iteration, the demand unit to be processed next is picked from the MinDistance heap using the extract-min operation.

Let the result of extract-min operation be a demand unit ${d_i}$ and its ``most suitable'' service center be $s_j$. $d_i$ is allotted to $s_j$. Following this, the resultant allotment (let it be $\mathcal{A}2$) is updated according to following two cases: (a) $s_j$ had free capacity for the allotment or, (b) $s_j$ was already full. 

\emph{Case(a):} we increment the objective function ($\Delta$) by the amount $\mathcal{CM}(d_i,s_j)$. Following this, we remove any negative cycle (from $\Gamma(\mathcal{A}2)$) that could have been created by the recent allotment (using Algorithm \ref{algrnegc}).

\emph{Case(b):} we increment the objective function ($\Delta$) by the amount $\mathcal{CM}(d_i,s_j)+q_{s_i}()$. Here, $q_{s_i}()$ is the penalty function of service center $s_i$. Following this, we first remove any negative cycle that could have been created by the recent allotment (using Algorithm \ref{algrnegc}). Let $\mathcal{A}3$ be the resulting allotment. After this, we remove the most negative path which originates from $s_i$ (called as anchor vertex in the Algorithm) in $\Gamma(\mathcal{A}3)$. Algorithm \ref{algrnegp} is used for this purpose.

\begin{algorithm}[ht]
\caption{Allotment Subspace Re-adjustment based approach for LBDD}
\label{mainalgo}
\begin{flushleft}
\noindent \textbf{Input:} (a) Demand-service cost matrix $\mathcal{CM}$; (b) A set of service centers $S$. For each service center $s_i$, we have its capacity $c_{s_i}$ and penalty function $q_{s_i}()$; (c) A set of demand vertices  $D$; \\
\noindent \textbf{Output} A mapping from demand units to the service centers.
\end{flushleft}
\begin{algorithmic}[1]
\STATE For each $d_i \in D$ determine the ``most suitable'' service center and create a MinDistance heap with all $<$demand unit-best service center$>$ pairs
\STATE Initialize objective function $\Delta \leftarrow 0$
\STATE Initialize current allotment $i = 0$; $\mathcal{A}1 \leftarrow$ $NULL$; Initialize $\Gamma DS$     /*{\em Precondition: $\Gamma(\mathcal{A}1)$ has no negative cost cycles }*/
\WHILE {MinDistance heap is not empty}   
\STATE $i++$; $<{d_i}$, $s_j>$ $\leftarrow$ extract-min on MinDistance heap
\STATE $\mathcal{A}2 \leftarrow$ Change $\mathcal{A}1$ by allotting $d_i$ to $s_j$. Update $\Gamma DS$ 
\STATE $\mathcal{A}3 \leftarrow$ Negative Cycle refinement($\mathcal{A}2$,$s_j$) (using Algo \ref{algrnegc}) \newline  /*{\em Invariant: $\mathcal{A}3$ has no negative cost cycles }*/
\IF {\textit{$s_j$} has vacancy}
\STATE Decrement capacity of $s_j$
\STATE $\Delta$ $\leftarrow$ $\Delta + \mathcal{CM}(d_i,s_j)$  /*Increase objective function */
\STATE $\mathcal{A}1$ $ \leftarrow$ $\mathcal{A}3$
\ELSE /*Allotment leads to penalty at $s_j$ */
\STATE $\Delta$ $\leftarrow$ $\Delta + \mathcal{CM}(d_i,s_j) + q_{s_j}()$  
\STATE $\mathcal{A}4$ $\leftarrow$ Negative Path refinement($\mathcal{A}3$,$s_j$) (using Algo \ref{algrnegp}) \newline  /*{\em Invariant: $\Gamma(\mathcal{A}4)$ has no negative cost cycles }*/ 
\STATE $\mathcal{A}1$ $\leftarrow \mathcal{A}4$ 
\ENDIF
\ENDWHILE
\end{algorithmic}
\end{algorithm} 

\subsection{Para-ASRAL Algorithm}
\label{pasral}
This section presents our proposed Para-ASRAL algorithm which uses GPU and CPU cores to speed-up the execution. We achieve this speed-up through the following: (a) using a parallel Bellman Ford algorithm inside our negative cycle refinement and negative path refinement procedures; (b) use threads to update $\Gamma DS$ during the sub-space readjustment. We now give present details of these techniques.

%By using Parallel Perform Cascade algorithm for finding cascade paths in allotment subspace multigraph.

\subsubsection{\textbf{Parallel Bellman Ford algorithm}}
\label{negpath}
Both negative cycle refinement and negative path refinement procedures use the Bellman ford algorithm to compute the lowest cost path $P_{min}$ in $NCyc\Gamma()$ and $NPath\Gamma()$ respectively. However, this process is time consuming given its complexity in our case. More precisely, its complexity turns to out to be $O(k^3)$ in our case when we set $|E|=k^2$ and $|v|=k$ in ($O(|E||V|)$ (time complexity of the Bellman ford algorithm).

We implement a parallel version of the Bellman Ford algorithm where the inner loop of the Bellman's algorithm (which loops over all the edges) is parallelized using a GPU. We implemented this parallel version using JCuda \cite{jcuda} which is a java binding of the CUDA Runtime libraries and allows us to run java code on GPU. A detailed pseudo-code of the procedure is given in Algorithm \ref{bfalgo}.

%we need to find the most negative path $P_{min}$ in the allotment subspace multigraph. We can use the Bellman ford algorithm (time complexity: O($n^3$), where n is the number of nodes in a graph) for it. Because of the high running time, we implement a parallel bellman ford algorithm using JCuda, which reduces the running time of the ASRAL algorithm. JCuda is a java binding of the CUDA Runtime libraries, which help us to run our java code on GPU \cite{jcuda}. It allocates java data structures (Host memory on CPU) to Device memory on GPU, then loads modules and executes the kernels on GPU from a Java application. A detailed pseudo-code of the procedure is given in Algorithm \ref{bfalgo}.

\begin{algorithm}[ht]
\caption{Parallel Bellman ford algorithm for finding negative path}
\label{bfalgo}
\begin{flushleft}
\noindent \textbf{Input:} (a) An adjacency matrix representation of $NCyc\Gamma()$ (or $NPath\Gamma()$); (b) A distance array of size $k$; (c) A parent array of size $k$; (d) A source vertex (anchor vertex) \\
\noindent \textbf{Output:} It returns the output in distance and parent array. The distance array contains the distance of shortest path from source to all other nodes of the given graph. The parent array contains the parent node of shortest path from source. 
\end{flushleft}
\begin{algorithmic}[1]
\STATE Initialize the distance array. Put $\infty$ for all nodes except the source node and put 0 for source node.
\STATE Initialize the parent array. Put $-1$ for all nodes. /*{Both arrays are on CPU i.e. Host memory}*/
\STATE Allocate memory for three arrays (matrix representation of $NCyc\Gamma()$ or $NPath\Gamma()$, distance, and parent) on GPU with the same size as host memory using JCuda. /*{i.e. device memory}*/
\State Transfer host memory data to device memory (i.e., GPU) using JCuda
\FOR {each node of the graph}
        \STATE Relax all edges using GPU (inner loop of Bellman ford)
        \STATE Synchronize GPU /*{wait for the completion}*/
\ENDFOR
\STATE Transfer device memory data to host memory using JCuda
\STATE Free all GPU memory
\end{algorithmic}  
\end{algorithm}

Parallel Bellman ford algorithm takes the following input: 1) An adjacency matrix representation of $NCyc\Gamma()$ (or $NPath\Gamma()$), 2) distance array,  3) parent array and 4) a source node (anchor vertex). We allocate space for input data structures on the CPU main memory (referred to as host memory) and initialize them. Using JCuda, we also allocate space for our data structures on the GPU memory (referred to as the device memory). Data containing initial values of our data structures is transferred from the host memory to the device memory. After the computation, the final results are transferred back to the host memory. 

Inside the Bellman ford algorithm, output of $(i+1)^{th}$ iteration depends on $i^{th}$ iteration. However, the relaxation of edges inside each iteration (which is the inner loop of Bellman's algorithm) is an independent operation. This allows us to take advantage of the GPU cores. We run the relaxation of all edges (i.e., the inner loop of Bellman) part on GPU using JCuda. Because of the dependency across the iterations of the outer loop, we need to synchronize the GPU after each iteration otherwise, we will not get accurate results. 

After all the iterations are over (i.e., the for loop at line 5 in Algo \ref{bfalgo}), we transfer the data from the device memory to the host memory and free all the device memory.

\subsubsection{\textbf{Updating $\Gamma D S$ in parallel}}
\label{cascade}
After finding the most negative path $P_{min}$ in both negative cycle refinement (Line 3 in Algo \ref{algrnegc}) and negative path refinement procedure (line 3 in Algo \ref{algrnegp}), the ASRAL algorithm adjusts the allotment and $\Gamma D S$ according to $P_{min}$ path. Each edge $e \in P_{min}$ represents as a tuple $<s_{old},dn,s_{new}>$ where we transfer a demand unit $dn$ from service center $s_{old}$ to $s_{new}$. This process involves updates to both the allotment and the $\Gamma D S$. 

% and the cost of edge $e$ is defined as $\mathcal{CM}(dn,s_p) - \mathcal{CM}(dn,s_q)$.
 
Updating allotment according $P_{min}$ involves involves moving the demand units across service centers according the edges $<s_{old},dn,s_{new}>$ present in the path $P_{min}$. As an example consider the following sample $P_{min}$ with two edges: $<s_1,x,s_2>$;$<s_2,y,s_3>$. Here, demand unit $x$ would be transferred from $s_1$ to $s_2$ and, demand unit $y$ would be transferred from $s_2$ to $s_3$. Clearly, these two operations are trivially independent of each other and can be done in parallel. However, this aspect was not explored in this work with the intention of using threads in the slightly more involved step of updating $\Gamma D S$ (described next).

After the allotment is updated, we need to also update $\Gamma D S$ for the service centers involved in $P_{min}$. Consider again the previously mentioned sample $P_{min}$ with two edges: $<s_1,x,s_2>$;$<s_2,y,s_3>$. Following operations are needed for the first edge, i.e., $x$ moving from $s_1$ to $s_2$. Demand node $x$ needs to be removed from $k-1$ heaps corresponding to $s_1$ in the $\Gamma D S$ and, then it needs to be added to $k-1$ heaps corresponding to $s_2$. Note that $k$ is the number service centers in the given problem. All these update operations to these $2k-2$ heaps are independent to each other and are done in parallel using threads. Similar process is repeated for all the other edges in $P_{min}$.

\section{Analytical Evaluation}
\label{ana}
\subsection{Theoretical Analysis}

%{\em (3) Dynamically maintaining an allotment under addition and deletion of  demand nodes.} Since Algorithm \ref{mainalgo} works by adding demand units one at a time, adding a new demand unit is taken care of. Using our data structures, the dropping of a demand unit can be handled in a very similar manner, at the same cost, and preserving the property of the allotment (optimal for the capacity constrained version, and negative-cycle free for the penalty based version).  This is proved in Theorem \ref{dynamic}.
   
In this subsection, we now prove key Lemmas that establish the correctness of Algorithm \ref{mainalgo}. The correctness depends on the correctness of the in-variants after Steps 7 and 14.
      
\textbf{For Step 7:} We detect the most negative cycle in $\Gamma(\mathcal{A}2)$, and if that cycle has a cost less than zero, modify the allotment to remove the cycle. The resulting allotment is $\mathcal{A}3$, and we show that $\Gamma(\mathcal{A}3)$ does not contain any negative cost cycles.
     
\textbf{For Step 14:} We detect the most negative path, originating in $s_j$, from $\Gamma(\mathcal{A}3)$, and if that path has a cost less than zero, modify the allotment to remove the path. The resulting allotment is $\mathcal{A}4$ and we show that $\Gamma(\mathcal{A})4$ does not contain any negative cost cycles. 
 
\begin{lemma}
Consider the allotment $\mathcal{A}1$ at the start of Step 6 of the Algorithm \ref{mainalgo}. If $\Gamma(\mathcal{A}1)$ does not have any negative cost cycles then $\Gamma(\mathcal{A}3)$ (at the end of step 7) does not have any negative cost cycles.
\label{mainlemma}
\end{lemma}    

\noindent \textbf{Proof Strategy:} The only new edges added to $\Gamma(\mathcal{A}1)$ (to create $\Gamma(\mathcal{A}2)$) are due to addition of $d_i$ and are all leading out of $s_j$. These edges move $d_i$ to each of the other service centers. Any negative cycle introduced in $\Gamma(\mathcal{A}2)$ must therefore contain an edge leading out of $s_j$, otherwise it would imply that $\Gamma(\mathcal{A}1)$ had negative cycles at the start of Step 5. Let $P_{min}$ be the path determined by the negative cycle refinement algorithm (Algorithm \ref{algrnegc}) between from $w_{out}$ (out copy of $s_j$) to $w_{in}$ (in copy of $s_j$). 

Let $C_{min}$ be the cycle formed in $\Gamma(\mathcal{A}3)$ by considering the edges in $P_{min}$ and collapsing the ``in'' and ``out'' copies of the $s_j$ ($w_{in}$ and $w_{out}$ in Algorithm \ref{algrnegc}). Note that this cycle is trivially guaranteed to exist as $NCyc(\Gamma(\mathcal{A}), s_j)$ is essentially created by breaking $s_j$ into its ``in'' and ``out'' copies. Also note that if cost of $C_{min}$ is greater than zero, then the statement is trivially correct. Thus, we assume that cost of $C_{min}$ is less than zero. 

Overall, we prove our lemma using proof by contradiction, i.e., if  $\Gamma(\mathcal{A}3)$ has a negative cycle (say $C_{neg}$) then there must have been a negative cycle in  $\Gamma(\mathcal{A}1)$. We shall first present some essential concepts (and an example) before presenting details of the proof. 

Our proof uses edges from $C_{min}$ and $C_{neg}$ to find a cycle, $C_{cont}$, in $\Gamma(\mathcal{A}1)$ with total cost less than zero. $C_{min}$ consists of edges from  $\Gamma(\mathcal{A}2)$ and $C_{neg}$ consists of edges from $\Gamma(\mathcal{A}3)$. We perform a two-step process to find $C_{cont}$: 
 
(1) {\em Nullify edge dependencies:} Two edges, $e_1$ and $e_2$, are termed as {\bf dependent} if they have the form $e_1 = (s_p, s_q, d_x)$ and $e_2 = (s_q, s_r, d_x)$, i.e., $e_1$  transfers  $d_x$ from $s_p$ to $s_q$ and $e_2$ transfers the same demand unit $d_x$ from $s_q$ to $s_r$. Note that a dependency can happen only when $e_1$   and $e_2$ belong to different allotment subspace multigraphs. This dependency is {\bf nullified} by removing $e_1$ and $e_2$, and adding the edge $e_3 = (s_p, s_r, d_x)$. Note that the cost of $e_3$ is the sum of the costs of $e_1$ and $e_2$, i.e., {\em the total cost remains unchanged when dependencies are nullified}. For finding $C_{cont}$ we take the edges in $C_{min}$ and $C_{neg}$ and nullify all dependencies. Once dependencies are nullified, the resulting set of edges has the property that any demand unit, $d_x$ ($x\neq i$), is moved at most once, and this movement is from the service center node to which $d_x$ was assigned in $\mathcal{A}1$. Therefore all these edges are from $\Gamma(\mathcal{A}1)$, except for the new edge in $C_{min}$ that moves $d_i$ out of $s_j$.

(2) {\em Identify and exclude  a cycle containing the new edge that moves $d_i$ out of $s_j$.} If we exclude the edges on this cycle,  all the remaining edges must be from $\Gamma(\mathcal{A}1)$. As we shall see, the negative cycle, $C_{cont}$, can be found among these remaining edges.

Figure \ref{lemfig} shows an example of our proof strategy. $S^* $ is the service center to which the new demand unit, $d_i$, is added. In the Figure, the edges $(S1, Sa)$ and $(S4, Sd)$ in $C_{neg}$ are dependent, respectively, on edges $(S^*, S1)$ and $(S3, S4)$ belonging to $C_{min}$. While nullifying these dependencies, we add the edges $(S^*, Sa)$ and $(S3, Sd)$. This leaves us with two independent cycles: an {\em inner cycle} $S1 \rightarrow S2 \rightarrow S3 \rightarrow Sd \rightarrow S1$, and an {\em outer cycle}, $S^* \rightarrow Sa \rightarrow Sb \rightarrow Sc \rightarrow S4 \rightarrow S5 \rightarrow S^*$. The  edge $(S^*, S1)$, which moved $d_i$ to  $S1$ is removed and the   edge  $(S^*, Sa)$ is introduced, which moves $d_i$ to  $Sa$.  All the edges in the inner cycle $S1 \rightarrow S2 \rightarrow S3 \rightarrow Sd \rightarrow S1$ belong to $\Gamma(\mathcal{A}1)$. The outer cycle $S^* \rightarrow Sa \rightarrow Sb \rightarrow Sc \rightarrow S4 \rightarrow S5 \rightarrow S^*$ contains the  only edge that moves $d_i$, and all these edges are present in $\Gamma(\mathcal{A}2)$. 

\begin{figure}[h]
\begin{center}
\includegraphics[width=60mm]{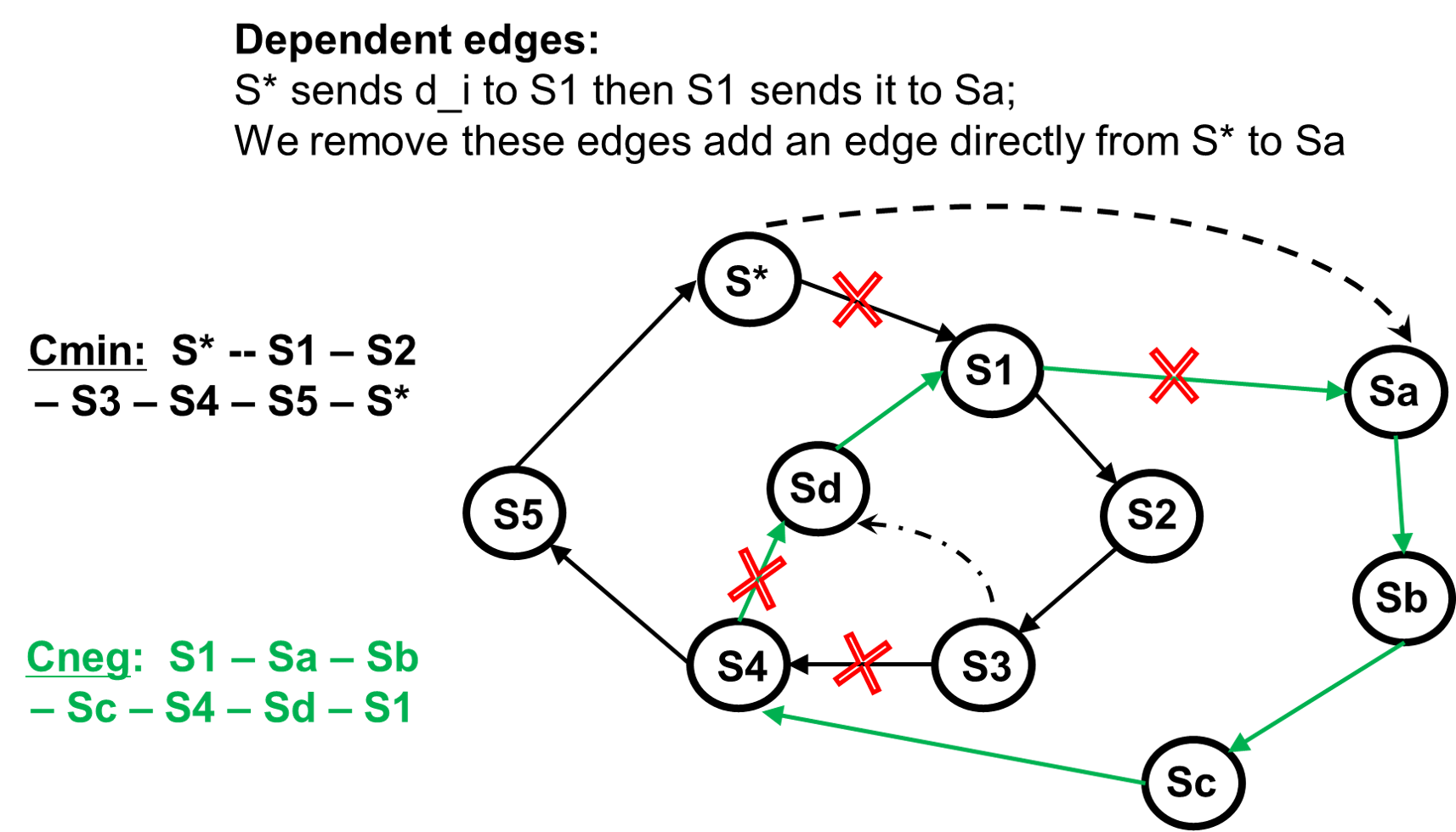}
\caption{Example illustrating $C_{min}$ and $C_{neg}$ cycles for proof of Lemma \ref{mainlemma} (best is color).}
\label{lemfig}
\end{center}
\vspace{-5mm}
\end{figure}

Therefore, since $C_{min}$ was the least cost cycle in  $\Gamma(\mathcal{A}2)$,  the outer cycle must have cost greater than or equal to that of $C_{min}$. Note that the dependency nullification process preserves the total cost, and therefore the total cost of the inner cycle, $S1 \rightarrow S2 \rightarrow S3 \rightarrow Sd \rightarrow S1$, has to be less than or equal to the cost of $C_{neg}$, i.e., it has to be negative. Thus the inner cycle is the $C_{cont}$ we are looking for.
 
\begin{proof}
{\bf Proof details of Lemma \ref{mainlemma}. }  Note that $C_{min}$ contains exactly one  new edge leading out of $s_j$.  All pairs of edges in $C_{min}$ are independent, as are all pairs of edges in $C_{neg}$. Consider the digraph $H(V,E)$, defined by $C_{min} \bigcup C_{neg}$. Since $H$ is a union of two cycles, it satisfies the property that $\forall v \in V, indegree(v) = outdegree(v)$. $H$ will have pairs  of edges, $(e_1 \in C_{min}, e_2 \in C_{neg})$, that are dependent. Let $H'(V,E')$ be the resulting digraph, after we nullify the dependencies in $H$. For any dependency, the vertex, $v_q$, corresponding to service center $s_q$, must belong to both cycles, i.e., $indegree(v_q) = outdegree (v_q) = 2$. After nullifying the dependency, we have $indegree(v_q)$ $= outdegree (v_q) = 1$.  Hence $H'$ also satisfies the property that $\forall v \in V, indegree(v) = outdegree(v)$, i.e. $H'$ admits a decomposition into directed cycles. Note the following:\\   
\noindent{\bf Claim 1:} {\em All the edges in $H'$ belong to $\Gamma(\mathcal{A}2)$.} The only way this property can be violated is when $C_{neg}$ contains an edge that depends on an edge in $C_{min}$. Since all dependencies are nullified, the claim holds. \\ 
{\bf Claim 2:} {\em All the edges in $H'$ belong to $\Gamma(\mathcal{A}1)$, except for the new edge leading out of $s_j$.} The edges in $C_{min}$ are from  $\Gamma(\mathcal{A}1)$, except for the new edge. Any edge in $C_{neg}$ that is not in $\Gamma(\mathcal{A}1)$, is dependent on some edge in $C_{min}$; all such edges are replaced by edges from $\Gamma(\mathcal{A}1)$ when we nullify dependencies. This leaves only one new edge leading out of $s_j$ that is not in $\Gamma(\mathcal{A}1)$.   \\
{\bf Claim 3:} {\em Total cost of all edges in $H'$ is less than the total cost of the edges in $C_{min}$.} The edges in $C_{neg}$ have total cost less than zero, and dependency nullification preserves the cost.\\ 
To complete the argument, find the cycle $C1$ of least cost that includes the new edge leading out of $s_j$. This cycle must have cost greater than or equal to cost of $C_{min}$, since $C_{min}$ was the cycle of least cost in  $\Gamma(\mathcal{A}2)$. Therefore, $H' - C1$ has cost less than zero, i.e., decomposing it into directed cycles yields at least one cycle of negative cost. Since all the edges in $H' - C1$ belong to $\Gamma(\mathcal{A}1)$, this negative cycle existed in $\Gamma(\mathcal{A}1)$.  
\end{proof}

\begin{lemma}
$\Gamma(\mathcal{A}4)$  does not have any negative cost cycles at the end of step 14 in Algorithm \ref{mainalgo}
\label{mainlemma2}
\end{lemma}

\begin{proof}
Let $P_{min}$ be the path of least cost in $NPath\Gamma(\mathcal{A}3,s*)$ between the ``in" and ``out" copies ($s*_{in}$ and $s*_{out}$) of $s*$. $P_{min}$ comprises of demand unit transfer edges upto a service center node (say) $s_l$ followed by penalty transfer edge from $s_l$ to the $s*_{in}$. Let $C*$ be a cycle created by collapsing $s*_{in}$ and $s*_{out}$ in $P_{min}$. All the demand unit transfer edges in $C*$ belong to $\Gamma(\mathcal{A}3)$. By way of contradiction, assume that there is a negative cycle, $C_{neg}$ in $\Gamma(\mathcal{A}4)$. Define the digraph $H(V,E)$, as $C* \bigcup C_{neg}$, and let $H'$ be the resulting digraph when we eliminate dependencies from $H$. All the edges in $H'$, other than the penalty transfer edge, belong to $\Gamma(\mathcal{A}3)$. Let $C1$ be the cycle of least cost in $H'$ that contains the penalty transfer edge from $s_l$ to $s*$. $C1$ = a path from $s*$ to $s_l$ + the penalty transfer edge from $s_l$ to $s*$. Similar to the proof of Lemma \ref{mainlemma}, we claim that the digraph  $H' - C1$ contains edges only from $\Gamma(\mathcal{A}3)$, has total cost less than zero, and admits a decomposition into directed cycles.  Therefore it must contain a negative cycle, which also existed in  $\Gamma(\mathcal{A}3)$. $^{\boxed{}}$
\end{proof}

%\begin{theorem}
% There exists a fully dynamic data structure that uses $O(nk)$ space and allows us to insert and remove demand nodes in $O(k^3 + k^2\log n)$  time, where $n$ is the number of demand nodes and $k$ is the number of service centers. The data structure  maintains an optimal allotment for the capacity constrained version  and maintains an allotment that is negative cycle free for the penalty-based version.  
%\label{dynamic}
%\end{theorem} 

%Detailed Proof in Appendix \ref{detproofs}. 

\subsection{Space and Time Complexity Analysis}
\begin{theorem}
 Given an LBDD instance, Algorithm \ref{mainalgo} runs in time $O(nk^3+ nk^2\log n)$ time using $O(nk)$ space, where $n$ is the number of demand nodes and $k$ is the number of service centers.
\label{asralcompl}
\end{theorem} 

Detailed Proof in Appendix \ref{detproofs}.

\section{Experimental Analysis}
\label{exp}
\noindent \textbf{Datasets:} We used the road network of New Delhi, India in our experiments. This road network contained 65,000 nodes and was obtained from OpenStreetMaps [www.openstreetmap.org]. In our experiments, we randomly chose some vertices from the graph and made them service centers. The remaining vertices were designated as demand vertices (each with unit demand). Shortest distance (on the road network) represented the cost assigning a demand unit to a service center. We pre-computed the shortest path distances between all pairs of demand units and service centers.

The number of service centers was varied in the experiments according to ratio $n:k$ (called as service center ratio). Here, $n$ represents the number of demand units and $k$ represents the number of service centers. A ratio of $600:1$ implies that we have one service center for every $600$ demand vertices. For a given number of service centers and demand, the total capacity across all service centers is given by $\theta*n$. This total capacity is distributed amongst $k$ service centers randomly. For sake of easy interpretation of results, instead of generating monotonic functions, we assume that each service center has a penalty value (chosen randomly) inside a range. This range was varied in experiments.  

 Algorithms were implemented in Java 1.8. We conducted our experiments on a Ubuntu machine which had a NVIDIA Tesla P100 GPU (3584 CUDA cores) and Intel Xeon Silver 4110 CPU 2.10GHz processor. Our ASRAL algorithm used around 16GB of RAM in the main memory. Para-ASRAL used an additional 8GB on the GPU memory. \\

\begin{table}[t]
    \centering
    \begin{tabular}{|m{1cm}|m{3cm}|m{3cm}|} 
        \hline 
        \textbf{Ratio} & \textbf{\# Demand units} & \textbf{\# Service center}\\
        \hline 
        500:1   & 65771	& 131 \\ \hline
        600:1	& 65793	& 109 \\ \hline
        700:1	& 65808	& 94 \\ \hline
        800:1	& 65820	& 82 \\ \hline
        900:1	& 65829	& 73 \\ \hline
    \end{tabular}
    \caption{Number of Demand units and Service centers for different service center ratio values.}
    \label{ratio_dis}
    \vspace{-5mm}
\end{table}

\noindent \textbf{Candidate algorithms:} Following $4$ algorithms were compared in experiments. 

\begin{enumerate} 
\item \textbf{Our proposed ASRAL algorithm and Para-ASRAL algorithm.} 

\item \textbf{LoRaL algorithm \cite{loral}} We ran LoRaL by setting its ``exploration parameter'' (value in $k$ in \cite{loral}) to the number of service centers (i.e., $\lvert S \rvert$ ). This means that LoRaL would explore $\lvert S \rvert$  number of potential paths for re-organization (in the current partial solution) and then choose the best. Note that this is the maximum amount of exploration that LoRaL algorithm can perform on any LBDD problem instance. In other words, \textbf{the solution quality of LoRaL reported in this section is the best that can be attained by LoRaL}.

\item \textbf{Min-Cost Bipartite Matching:} Procedure given in \cite{NET20477}.
\end{enumerate}

\noindent \textbf{Variable Parameters:} Following are details of the parameters varied in our experiments.
\begin{enumerate}
    \item  Total capacity parameter value $\theta$ was either 0.3 (low total capacity) or 0.7 (high capacity) in our experiments.
    \item  Penalty values of all the service centers were either chosen from the range [1 to 200] (low penalty values) or from the range [200 400] (high penalty values) in our experiments.
    \item Service center ratio ($n:k$): We choose the following five ratio values 500:1, 600:1, 700:1, 800:1, and 900:1. Table \ref{ratio_dis} describes the number of demand units and service centers in our dataset for different service center ratio values. 
\end{enumerate}
 
We present our results in two parts. First, we compare the running time and solution quality of ASRAL algorithm with alternative approaches (Min-Cost Bipartite Matching and LoRaL) Following this, we present the experimental analysis of para-ASRAL algorithms and compare it with ASRAL and LoRaL algorithm. Note that solution quality of ASRAL and para-ASRAL is the same.

\subsection{Comparing ASRAL with Related Work}
\noindent \textbf{Comparing ASRAL with Min-Cost Bipartite Matching:} Figure \ref{mcexp} illustrates these results.  Our results (Figure \ref{mcexp-a}) indicate that optimal algorithm for \textbf{min-cost bipartite matching does not scale beyond graphs containing 4000 vertices}. Even its main memory requirements grew very fast. It requires greater than 90GB RAM for a LBDD problem instance with 4000 demand vertices and 40 service centers. Due to these reasons, optimal algorithm could not be included in experiments on large problem instances (on road network of size 65000 nodes). 
    
\begin{figure}[ht]   
\vspace{-3mm}   
\begin{center}
\subfigure[Execution Time]{\label{mcexp-a}\includegraphics[width=55mm]{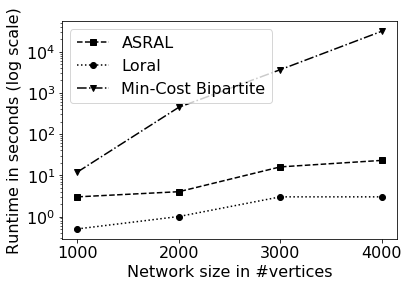}}
\subfigure[Objective Function Value]{\label{mcexp-b}\includegraphics[width=55mm]{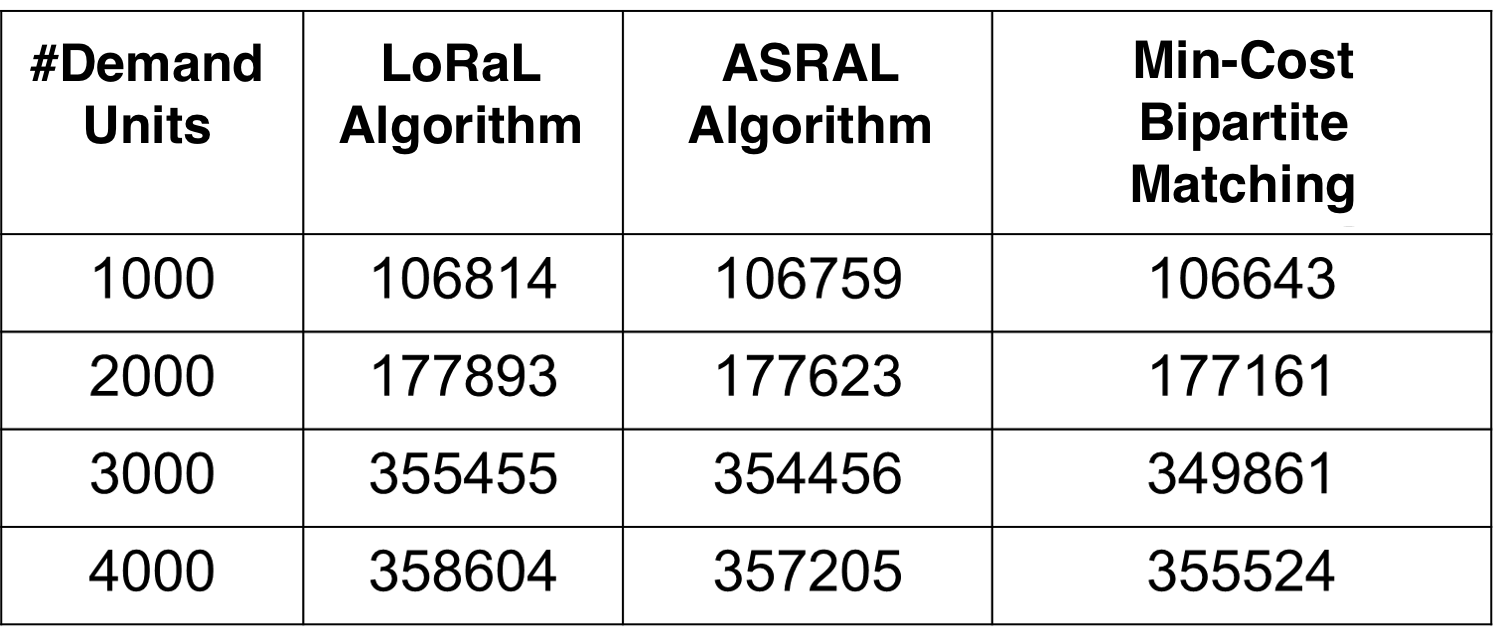}}
\vspace{-2mm}
\caption{ASRAL vs Min-Cost Bipartite Matching algorithm. $n:k$ was $300:1$}
\label{mcexp} 
\end{center}  
\vspace{-5mm} 
\end{figure}

\noindent \textbf{Comparing ASRAL with LoRaL:}
Figure \ref{Sol-asral} and Figure \ref{RT-asral} illustrates these results. \textbf{ASRAL obtains significantly lower values of objective function}. As the Figure \ref{Sol-asral} shows, our proposed ASRAL obtains a significantly better solution whose objective function value in some cases (see Figure \ref{S14}) is lower by \emph{half a million (or more) units}. In general, the difference is higher in two cases: (a) service centers have higher penalty values (between 200 -- 400) and/or; (b) number service centers is small (refer data for |D|:|S|=800:1 and 900:1 in Figure \ref{Sol-asral}). ASRAL does have more execution time than LoRaL (Figure \ref{RT-asral}), however, the key thing to note is that this increased execution time leads to significantly lower objective function values. ASRAL was able to scale up to large problem instances which the optimal algorithm for min-cost bi-partite matching was unable to do so.

\begin{figure}[ht]
     \centering
     \subfigure[Capacity = 0.3 and Penalty range from 1 to 200]{\label{S11}\includegraphics[width=0.48\textwidth,height=4cm]{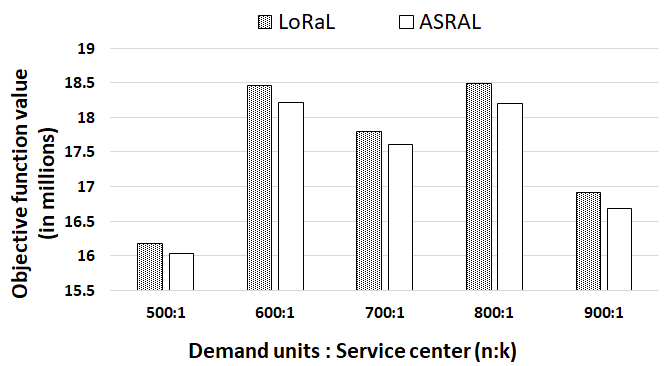}}
     \subfigure[Capacity = 0.3 and Penalty range from 200 to 400]{\label{S12}\includegraphics[width=0.48\textwidth,height=4cm]{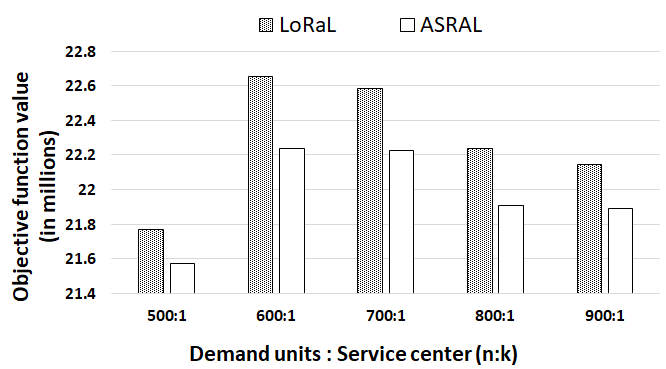}}
     \subfigure[Capacity = 0.7 and Penalty range from 1 to 200]{\label{S13}\includegraphics[width=0.48\textwidth,height=4cm]{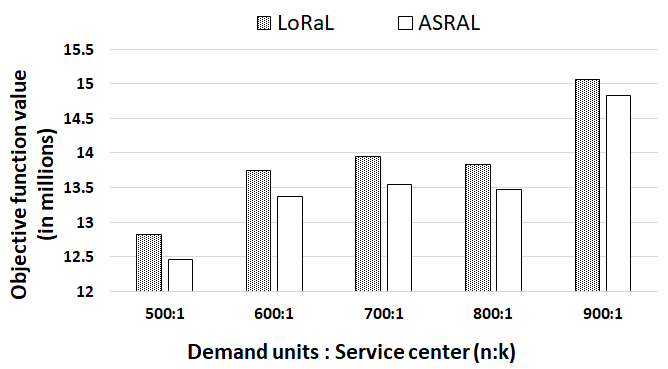}}
     \subfigure[Capacity = 0.7 and Penalty range from 200 to 400]{\label{S14}\includegraphics[width=0.48\textwidth,height=4cm]{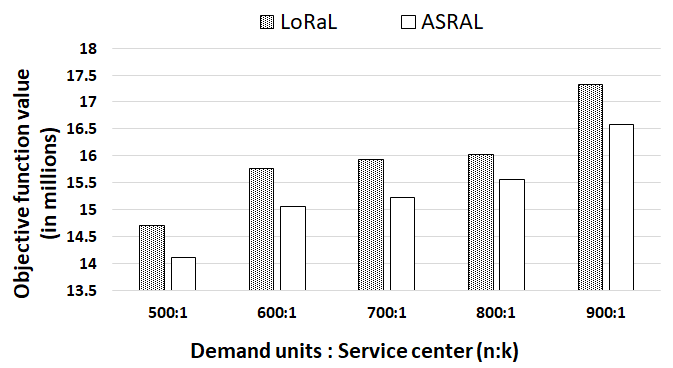}}
\caption{Final Objective function value of LoRaL and ASRAL algorithm}
\vskip 0pt
\label{Sol-asral}
\vspace{-2mm}
\end{figure}

\begin{figure}[ht]
     \centering
     \subfigure[Capacity = 0.3 and Penalty range from 1 to 200]{\label{R11}\includegraphics[width=0.48\textwidth,height=4cm]{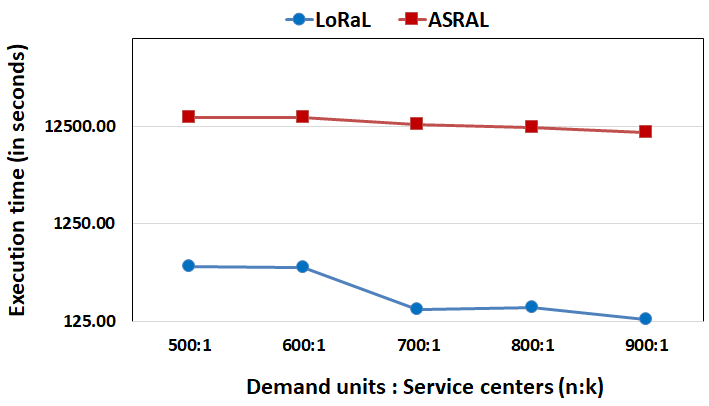}}
     \subfigure[Capacity = 0.3 and Penalty range from 200 to 400]{\label{R12}\includegraphics[width=0.48\textwidth,height=4cm]{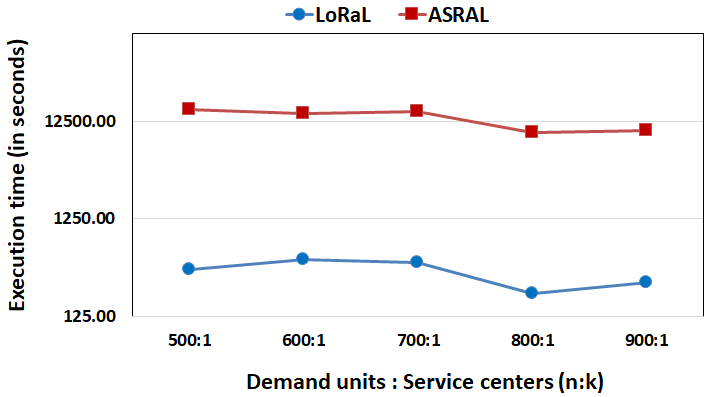}}
     \subfigure[Capacity = 0.7 and Penalty range from 1 to 200]{\label{R13}\includegraphics[width=0.48\textwidth,height=4cm]{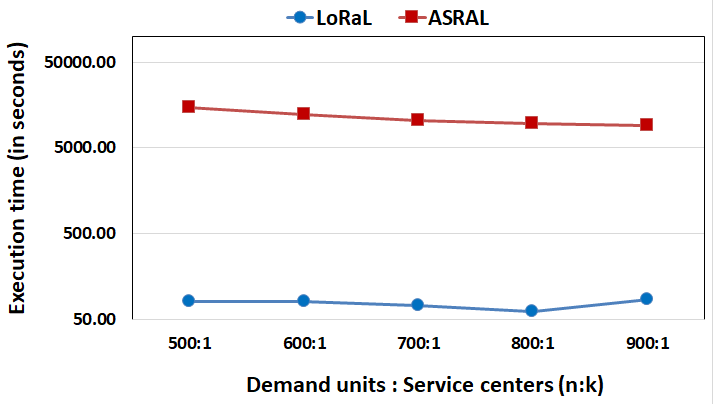}}
     \subfigure[Capacity = 0.7 and Penalty range from 200 to 400]{\label{R14}\includegraphics[width=0.48\textwidth,height=4cm]{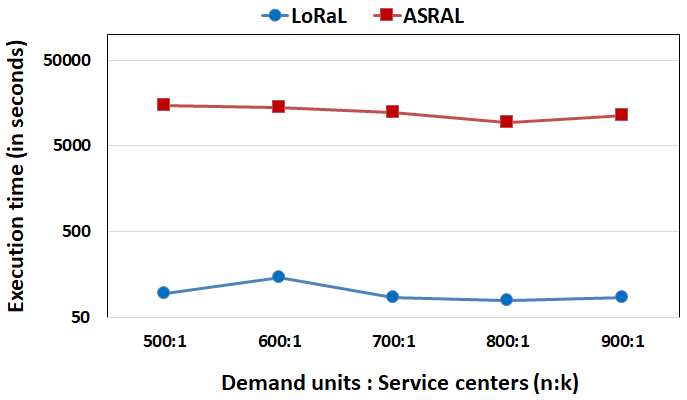}}
\caption{Execution time comparison of LoRaL and ASRAL algorithm (in log10 scale)}
\vskip 0pt
\label{RT-asral}
\vspace{-2mm}
\end{figure}

\noindent \textbf{Analyzing ASRAL running-time:} In Figure \ref{avgtime}, we show a typical break-up of execution time of ASRAL algorithm.  The figure shows that 74.15\% of the execution time was spent in updating $\Gamma DS$. Bellman ford algorithm took and 6.44\% of the execution time. And all other operations take up a total of 19.41\% time. The above observations motivate us to implement the parallel version of the ASRAL algorithm by using threads for updating $\Gamma DS$ and GPU cores for the bellman ford algorithm. These ideas were discussed in Section \ref{negpath} and Section \ref{cascade} respectively.
 
\subsection{Comparing ASRAL algorithm and Para-ASRAL algorithm} 
\textbf{Effecting of parallelizing Bellman Ford algorithm:} Figure \ref{beltime} compare the total time spent in the Bellman ford algorithm (used inside negative cycle and negative path refinement procedures) in ASRAL and Para-ASRAL algorithms. As mentioned earlier, Para-ASRAL uses GPU codes to parallelize the Bellman ford algorithm. Our results show that, we are able to reduce much of the running time of the Bellman ford algorithm using GPU cores. We also observe that reduction in running time is higher at lower service center ratios (e.g., values at 500:1 and 600:1 in Figure \ref{beltime}). This is due to the fact that the size of the allotment subspace mulitgraph decreases as we increase the service center ratio (in our dataset). For example, at 900:1  service center ratio, we would be having just 73 service centers (in our dataset). Whereas, we had 131 service centers at 500:1 ratio. 

\textbf{Effecting of parallelizing $\Gamma DS$ update:} Figure \ref{castime} illustrates the total time spent in updating the $\Gamma DS$ in the ASRAL and para-ASRAL algorithms. As mentioned earlier, para-ASRAL uses CPU cores to perform updates to $\Gamma DS$ in parallel. Our results indicate that we were to able to reduce the $\Gamma DS$ updation time by at-least 85\% across different values of total capacity, penalty and service center ratios.  

\textbf{Total speed-up achieved by Para-ASRAL algorithm:} In Figure \ref{tottime}, we compare the total running time of ASRAL, para-ASRAL and LoRaL algorithms. Our experimental results indicate that Para-ASRAL has at-least 78\% less running time than ASRAL algorithm across different values of total capacity, penalty and service center ratios. Running time of Para-ASRAL is much closer to LoRaL algorithm. 

\begin{figure}[H]
\centering
\includegraphics[width=80mm]{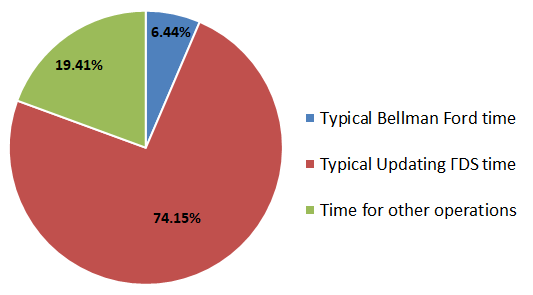}
\caption{Breakup of execution time of the ASRAL algorithm}
\label{avgtime}
\vspace{-2mm}
\end{figure}

\newpage 

\begin{figure}[H]
     \centering
     \subfigure[Capacity = 0.3 and Penalty range (1 200) ]{\label{bf11}\includegraphics[width=0.48\textwidth,height=4cm]{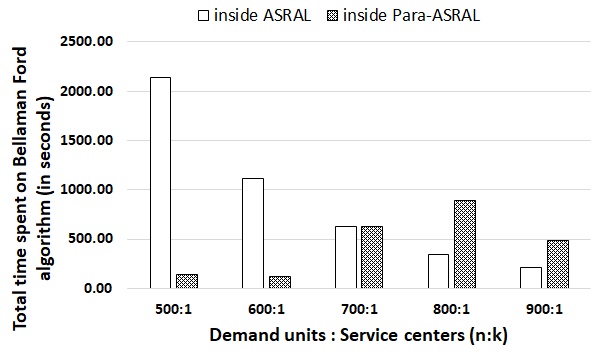}}
     \subfigure[Capacity = 0.3 and Penalty range (200  400)]{\label{bf12}\includegraphics[width=0.48\textwidth,height=4cm]{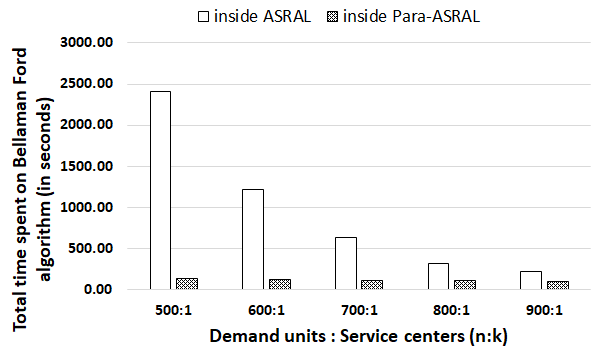}}
     \subfigure[Capacity = 0.7 and Penalty range (1  200)]{\label{bf13}\includegraphics[width=0.48\textwidth,height=4cm]{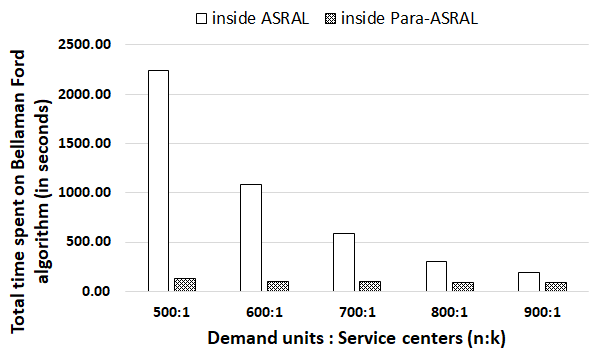}}
     \subfigure[Capacity = 0.7 and Penalty range (200  400)]{\label{bf14}\includegraphics[width=0.48\textwidth,height=4cm]{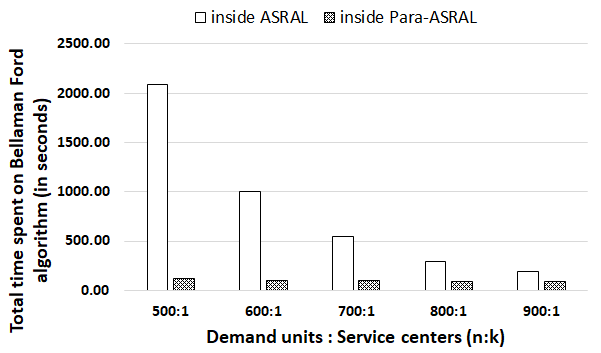}}
\caption{Total time spent on Bellman Ford algorithm inside ASRAL and para-ASRAL algorithm.}
\label{beltime}
\end{figure}

\newpage 

\begin{figure}[H]
     \centering
     \subfigure[Capacity = 0.3 and Penalty range (1 200)]{\label{f11}\includegraphics[width=0.48\textwidth,height=4cm]{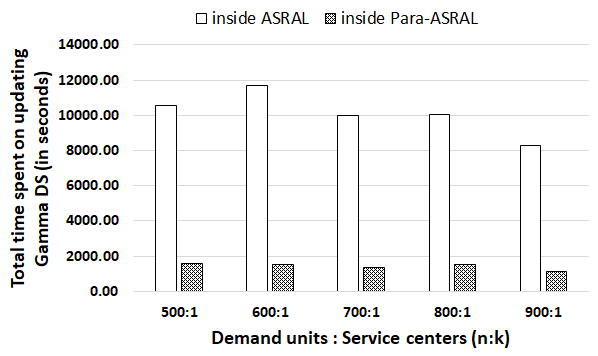}}
     \subfigure[Capacity = 0.3 and Penalty range (200  400)]{\label{f12}\includegraphics[width=0.48\textwidth,height=4cm]{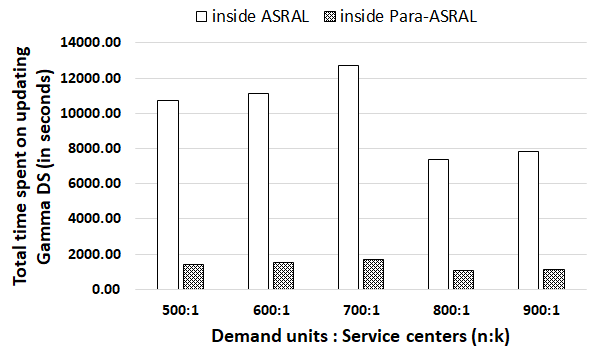}}
     \subfigure[Capacity = 0.7 and Penalty range (1  200)]{\label{f13}\includegraphics[width=0.48\textwidth,height=4cm]{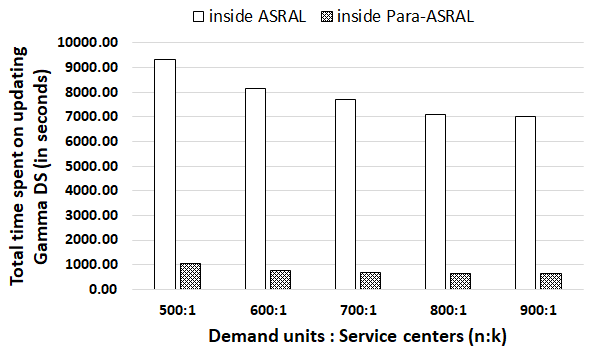}}
     \subfigure[Capacity = 0.7 and Penalty range (200  400)]{\label{f14}\includegraphics[width=0.48\textwidth,height=4cm]{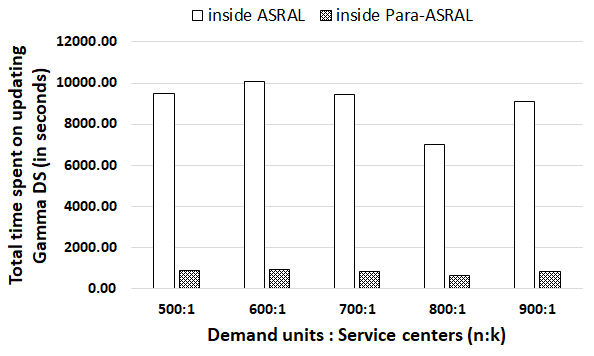}}
\caption{Total time spent on updating $\Gamma D S$ inside ASRAL and para-ASRAL algorithm}
\label{castime}
\end{figure}

\newpage 

\begin{figure}[H]
     \centering
     \subfigure[Capacity = 0.3 and Penalty range (1   200)]{\label{tt11}\includegraphics[width=0.48\textwidth,height=4cm]{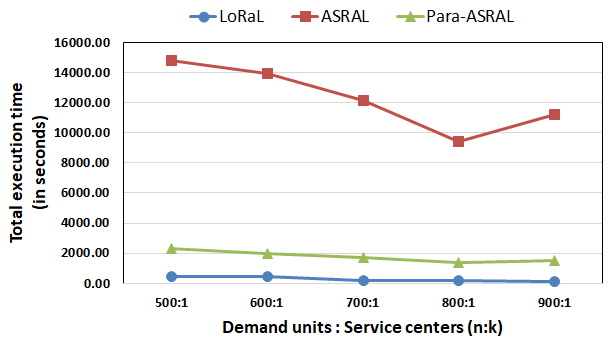}}
     \subfigure[Capacity = 0.3 and Penalty range (200   400)]{\label{tt12}\includegraphics[width=0.48\textwidth,height=4cm]{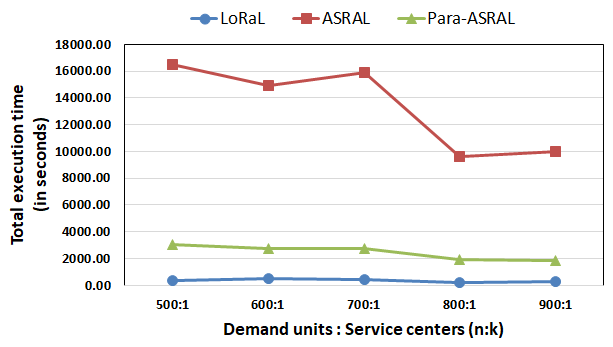}}
     \subfigure[Capacity = 0.7 and Penalty range  (1   200)]{\label{tt13}\includegraphics[width=0.48\textwidth,height=4cm]{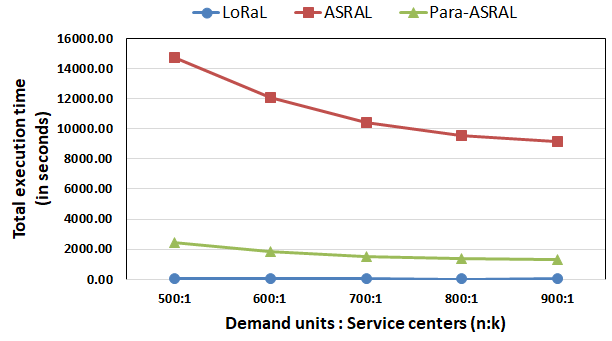}}
     \subfigure[Capacity = 0.7 and Penalty range (200   400)]{\label{tt14}\includegraphics[width=0.48\textwidth,height=4cm]{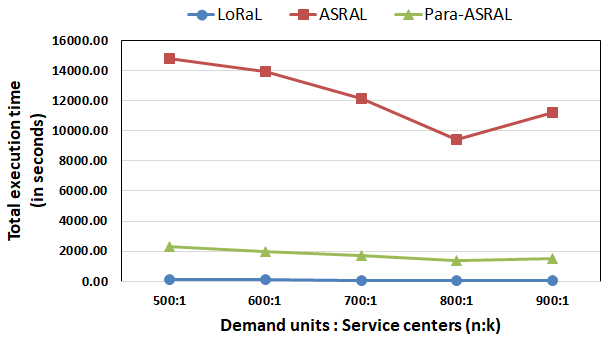}}
\caption{Total execution time of LoRaL, ASRAL and para-ASRAL algorithms}
\label{tottime}
\end{figure}

\newpage
   
\section{Conclusion}
\label{con}
Load Balanced Demand Distribution (LBDD) is a societally important problem. LBDD problem can be reduced to min-cost bipartite matching problem. However, optimal algorithms for min-cost bipartite matching cannot scale to large graphs. Our proposed ASRAL algorithm is able to scale-up while maintaining better solution quality over alternative approaches. We also develop a parallel implementation (Para-ASRAL) of the proposed ASRAL algorithm to improve its scalability. Para-ASRAL is significantly faster than ASRAL while maintaining the same solution quality.\\

\noindent \textbf{Acknowledgement:}  To be added later as the submission was double blind.\\

\noindent \textbf{Author Contributions:} To be added later as the submission was double blind.\\

\bibliographystyle{ACM-Reference-Format}
\bibliography{main}

\begin{appendices}

\section{Negative Path Refinement}
\label{negpalgo}
\begin{algorithm}[H]
\caption{Negative Path Refinement}   
\label{algrnegp}
\begin{flushleft}
\noindent \textbf{Input:} (a) Current allotment $\mathcal{A}$; (b) Anchor service center $SC_\theta$ \\
\noindent \textbf{Output} Allotment $\mathcal{F}$ after adjustment. Removes a negative path originating from $Sc_\theta$ in $\Gamma(\mathcal{A})$ of the input $\mathcal{A}$.
\end{flushleft}
\begin{algorithmic}[1]
\STATE Create NegPath allotment subspace graph $NPath\Gamma(\mathcal{A},SC_\theta)$ with $SC_\theta$ as the anchor vertex.
\STATE Let $w_{out}$ and $w_{in}$ be ``IN'' and ``OUT'' copies of the anchor vertex $SC_\theta$ in $minPath\Gamma(\mathcal{A})$.
\STATE Find lowest cost path $P_{min}$ in $NPath\Gamma(\mathcal{A})$ from $w_{out}$ to $w_{in}$
\IF{Cost of $P_{min} < 0$}
\FOR {each demand node transfer $<s_p,dn,s_q>$ $\in P_{min}$}
\STATE $\mathcal{F} \leftarrow$ Change $\mathcal{A}$ by removing $dn$ from $s_p$ and adding it to $s_q$. \State Update $\Gamma DS$
\STATE $\mathcal{A} \leftarrow \mathcal{F}$
\STATE $\Delta \leftarrow$ $\Delta~+$  Cost of corresponding edge in $P_{min}$ 
\ENDFOR
\STATE $\Delta \leftarrow$ $\Delta +$ Cost of last edge in $P_{min}$  /*This edge represents difference in penalties. 
\ENDIF
\STATE Return new allotment $\mathcal{B}$
\end{algorithmic}  
\end{algorithm}

%\section{Updating Boundary Vertices}
%\label{ubvalgo}

%\begin{algorithm}[htbp]
%\begin{flushleft}
%\noindent \textbf{Input:} (a) Demand node $d_i$; (b) Service center $SC_{new}$ where it is being allotted; (c) Current allotment $\mathcal{A}$ \\
%\noindent \textbf{Output} (a) Updated instance of data structure $BVHashMap$; (b) Updated allotment $\mathcal{A}^{\prime}$ 
%\end{flushleft}
%\begin{algorithmic}[1]
%\State $INneigh$ $\leftarrow$ All in-neighbors of $d_i$
%\If {$d_i$ already allotted to a service center $SC_{old}$}
%\State Update $\mathcal{A}$ by unalloting $d_i$ from $SC_{old}$
%\For {all in-neighbors $n_i \in INneigh(d_i)$}
%\If {$n_i$ is allotted to $SC_{old}$}
%\State Verify the boundary vertex property of $n_i$ according to Definition \ref{bdef}. /*It may now be a boundary vertex */
%\EndIf
%\EndFor 
%\EndIf
%\State Allot $d_i$ to $SC_{new}$ in $\mathcal{A}$, to get allotment $\mathcal{A}^{\prime}$
%\For {all in-neighbor $n_i \in INneigh(d_i)$}
%\If {$n_i$ is allotted to $SC_{new}$}
%\If {$n_i$ is a boundary vertex of $SC_{new}$}
%\State Verify the boundary vertex property of $n_i$ according to Definition \ref{bdef}. /*It may no longer be a boundary vertex */
%\EndIf
%\EndIf
%\EndFor
%\State Verify the boundary vertex property of $d_i$ according to Definition \ref{bdef}. /*It may be a boundary vertex of $SC_{new}$*/
%\State Update $\Gamma()$ 
%\end{algorithmic}
%\caption{Allot and Update Boundary Vertex Procedure}
%\label{algupdateBV}
%\end{algorithm}

\section{Optimal solution for LBDD under strict capacity constraints}
\label{optlbdd}

Our previous proposed algorithm for solving LBDD can be trivially adapted to give the \emph{optimal solution} when service centers are not allowed to be overloaded. We call this variation as \emph{strict-LBDD} problem. Following is formal definition of the problem.

\noindent \textbf{Strict-LBDD Problem}\\
 \textbf{Inputs}:\\			
			 (1) A set $S$ of $k$ service centers. \\
			 (2) A set $D$ of $n$ demand units. \\
			 (3) A demand-service cost matrix $\mathcal{CM}$. $\mathcal{CM}(i, j)$  contains  cost (positive real value) of allotting  demand unit $d_i $ to  service center $s_j$ ($\forall d_i \in D,~\forall s_j \in S$).\\ 
			 (4) Capacity $(c_{s_i})$ (integer) of each service center $s_i \in S$.  \\
 \textbf{Output}: A complete allotment which minimizes the following objective function. 
			\begin{multline}
			 \textit{\textbf{Minimize}}
			 \Bigg\{
			 \sum_{\substack{s_i \in Service \\Centers}} 	
				\Bigg\{
				\sum_{\substack{d_j \in Demand  \ unit\\ allotted \ to \  s_i}} 
					\mathcal{CM}(d_j,s_i)
					\Bigg\} \Bigg\}
			\label{eq3}   
			\end{multline}
  
\noindent \textbf{Constraints:} (i) A demand unit is allotted to only one service center; (ii) No service center can be allotted demand beyond its capacity. 

As can be imagined, if we have {\bf excess demand}, i.e., $n > \Sigma_{i=1}^{i=k} (c_{s_i})$, then it is not possible to generate a complete allotment. In such a case, we modify the problem instance as follows: 
\noindent (1) We add a new service center $s_{k + 1}$ of capacity $n - \Sigma_{i=1}^{i=k} (c_{s_i})$. 

\noindent (2) Let $max$ be the largest value in the given $\mathcal{CM}$. We modify $\mathcal{CM}$ by adding a  column for service center $s_{k+1}$ and set all the values in this column to $max+1$.

Note that this modification increases the value of the objective function by an amount equal to $(n-\Sigma_{i=1}^{i=k} (c_{s_i}))\times (max +1)$. Since this increase applies to all possible complete allotments, it does not affect the optimal solution.

\subsection{Characterizing the Optimal Solution in terms of $\Gamma()$}  
\label{olbdd}
For this purposes of characterizing the optimal solution of \emph{strict-LBDD}, we assume that all service centers are packed to capacity at all times using certain dummy demand units as any stage during the assignment process. These dummy demand units have zero assignment cost for all service centers. As a result, if $s_i$ has $p$ units of surplus capacity (for accommodating ``real'' demand units given in the input), then $\Gamma(\mathcal{A})$ would have $p$ edges of zero cost from $s_i$ to any $s_j$ due to the transfers of each of the $p$ dummy demand units.

As mentioned earlier, cycles of negative cost in $\Gamma(\mathcal{A})$ can help us improve the solution quality (i.e., lower the objective function values). The following theorem (Theorem \ref{optimalityThm}) formalizes the relationship between optimality of an \emph{strict-LBDD} solution and the presence of negative cycles. Similar results can be found in the literature (see for instance, Theorem 7.8 in \cite{graphbook}) but we provide the details for completeness. In our proofs, we employ the following proposition, which is common knowledge in Graph Theory.

\begin{proposition}
A directed graph, $G(V,E)$, admits a decomposition
into directed cycles if and only if, $\forall v \in V$, $indegree(v) = outdegree(v)$.  
\label{digraphdecomp}
\end{proposition}

\begin{theorem}
\label{optimalityThm}
Let $\mathcal{A}$ be an allotment for an instance of the \emph{strict-LBDD} problem, and $\Gamma(\mathcal{A})$ be the corresponding allotment subspace multigraph. $\mathcal{A}$ is an optimal solution if and only if $\Gamma(\mathcal{A})$ does not have any simple negative cost cycles.
\end{theorem}  
 
\begin{proof}
We assume in all our arguments that the cycles are simple. It is straightforward to show that if an allotment contains a non-simple negative cycle, it also contains a simple negative cycle.
 
Since the removal of a negative cost cycle reduces the total cost, it is obvious that an absence of negative cost cycles is necessary for optimality. We shall now prove that an absence of negative cost cycle is sufficient for optimality. %As discussed in the pre-processing step, we can assume without loss of generality that the total demand equals the total capacity.

Let $\mathcal{A}^{opt}$ be an optimal allotment, and  $\mathcal{A}$ be any allotment such that $\Gamma(\mathcal{A})$ does not contain  any negative cost cycles. We shall show that if $Cost(\mathcal{A}^{opt})$ is less than $Cost(\mathcal{A})$, then $\Gamma(\mathcal{A})$ must contain a negative cost cycle. 

We construct the {\em Difference Multigraph} for $\mathcal{A}$ and  $\mathcal{A}^{opt}$, denoted as  $DM(\mathcal{A}^{opt},  \mathcal{A})$ in the following way: \\
(1) For each service center $s_i, 1 \leq i \leq k$ (or $k+1$ in case of excess demand), add a vertex $v_i$ to $DM(\mathcal{A},  \mathcal{A}^{opt})$.\\
(2) Let $d_i, 1 \leq i \leq n,$ be any demand unit such that $\mathcal{A}$ assigns $d_i$ to service center $s_p$, and $\mathcal{A}^{opt}$ assigns $d_i$ to service center $s_q$, and $p \neq q$. $DM(\mathcal{A},  \mathcal{A}^{opt})$ contains an edge, $(v_p, v_q, d_i)$
of weight $\mathcal{CM}(d_i, s_q) - \mathcal{CM}(d_i, s_p)$. 
 
Note that $DM(\mathcal{A},  \mathcal{A}^{opt})$ represents the set of all demand unit movements needed to transform allotment $\mathcal{A}$ into  $\mathcal{A}^{opt}$. Therefore all these edges must belong to $\Gamma(\mathcal{A})$. 

{\bf Claim.} {\em We can ensure that the out-degree of each vertex $v_i$ in   $DM(\mathcal{A},  \mathcal{A}^{opt})$ is the same as in-degree of $v_i$ by appropriately adding edges that transfer dummy demand units.}\\ 
{\bf Proof of Claim.} In the construction of $DM(\mathcal{A},  \mathcal{A}^{opt})$, if the number of demand units moved out of $s_i$ is the same as the number of demand units moved in to $s_i$, we have nothing to prove. Let $\delta _i$ denote the difference between the number of demand units moved in to $s_i$ and  the number of demand units moved out of $s_i$. 
If $\delta _i > 0$, we will add $\delta _i$ outgoing edges of cost zero to $v_i$, for moving $\delta _i$ dummy demand units out of $s_i$. Note that since both $\mathcal{A}$ and  $\mathcal{A}^{opt}$ are valid allotments, these dummy demand units are guaranteed to exist.
Correspondingly, if $\delta _i < 0$, we can add $\delta _i$ incoming edges of cost zero to $v_i$. Since both $\mathcal{A}$ and  $\mathcal{A}^{opt}$ have the same number of demand units,  $\Sigma (\delta _i) = 0$. This means that total number of outgoing  dummy unit transfer arcs is equal to the total number of incoming  dummy unit transfer arcs. Since all these are of cost zero, it follows that we can arbitrarily connect the vertices with edges of weight zero to ensure that the claim is satisfied. {\bf end of proof of Claim}

By Proposition \ref{digraphdecomp} therefore, $DM(\mathcal{A},  \mathcal{A}^{opt})$ can   be partitioned into a set of edge-disjoint directed cycles.  Since $\mathcal{A}$ has greater cost than $\mathcal{A}^{opt}$, the sum of all the edge weights of $DM(\mathcal{A},  \mathcal{A}^{opt})$ must be less than zero, and therefore at least one of these cycles must have cost less than zero. All these edges must belong to   $\Gamma(\mathcal{A})$, i.e., $\Gamma(\mathcal{A})$ has a negative cycle.
\end{proof}

\subsection{Optimal algorithm for strict-LBDD problem}
Algorithm \ref{stlbddalgo} presents an optimal algorithm for the strict-LBDD problem. Optimiality of Algorithm \ref{stlbddalgo} can be trivially established using a combination of Lemma \ref{mainlemma} and Theorem \ref{optimalityThm}.

\begin{algorithm}[ht]
\caption{Subspace Re-adjustment based Approach for Strict-LBDD}
\label{stlbddalgo}
\begin{flushleft}
\noindent \textbf{Input:} (a) Demand-service cost matrix $\mathcal{CM}$; (b) A set $S$ of $k$ service centers. For each $s_i \in S$, we have its capacity $c_{s_i}$; (c) A set $D$ of $n$ demand units. \\
\noindent \textbf{Output:} An optimal allotment for the problem instance
\end{flushleft}
\begin{algorithmic}[1]
\STATE If total capacity is less than total demand then add dummy service center $s_{k+1}$ and modify $\mathcal{CM}$ by adding a column for $s_{k+1}$ and set all values to (max assignment cost in $\mathcal{CM}$ +1) 
\STATE Initialize $\mathcal{A}1$. Create dummy demand units and assign them to all service centers $s_i \in S$ at zero cost such that each service center is full in terms of capacity. 
\WHILE {there exists a visited demand unit}   
\STATE Pick an unvisited demand unit $d_i$. Mark $d_i$ as visited.
\STATE $X \leftarrow$ Set of all service centers that have residual capacity (i.e., they still some dummy demand nodes assigned to them).
\STATE Determine a service center $s_j$ from $X$ such that $\mathcal{CM}(d_i, s_j)$ is minimum.
\STATE $\mathcal{A}2 \leftarrow$ Change $\mathcal{A}1$ by replacing a dummy demand unit by $d_i$ in $s_j$ 
\STATE $\mathcal{A}3 \leftarrow$ Negative Cycle refinement($\mathcal{A}2$,$s_j$) (using Algo \ref{algrnegc})
\STATE $\mathcal{A}1 \leftarrow \mathcal{A}3$
\ENDWHILE
\STATE Return the resulting allotment after undoing the pre-processing step (i.e., remove any dummy demand units or remove the dummy service center $s_{k+1}$ (and its assigned demand units)).
\end{algorithmic}
\end{algorithm}

\section{Other Detailed Proofs}
\label{detproofs} 

\textbf{Proof of Theorem \ref{asralcompl}:} Determining the ``most suitable'' (using $\mathcal{CM}$) service center for each demand unit takes $O(nk)$ time. Following this, the ASRAL algorithm constructs the MinDistance heap, which takes another $O(n)$ time. Note that initially $\Gamma DS$ is empty. Therefore, we include its update costs during the course of the algorithm (i.e., between steps 4 and 17).

The main loop in ASRAL algorithm (Steps 4 to 17) executes $n$ times. Within this loop, {\em extract-min} operation (Step 5) takes $O(\log n)$ time. In Step 6, we allot $d_i$ to $s_j$ which requires updating $k-1$ heaps, and thus takes at most $O(k\log n)$ operations. Step 7 invokes the negative cycle refinement procedure (Algorithm \ref{algrnegc}). Within Algorithm \ref{algrnegc}, Step 1 constructs the $NCyc\Gamma()$ in $O(k^2)$ steps. Step3 finds $P_{min}$ in $O(k^3)$ steps, and the loop in Step 5 performs at most $k^2$ heap operations of $O(\log n)$ steps each. Algorithm \ref{algrnegc} therefore runs in $O(k^3+k^2\log n)$ steps.  After negative cycle refinement, ASRAL algorithm undertakes negative path refinement at Step 14 (using Algorithm \ref{algrnegp}). Similar to negative cycle refinement, negative path refinement (Algorithm \ref{algrnegp}) also takes $O(k^3+k^2\log n)$ steps. 

In summary, the entire while loop in Steps 4 to 17 of ASRAL can therefore be completed in $O(nk^3+nk^2\log n)$ time. Therefore, the total time complexity of ASRAL (Algorithm \ref{mainalgo}) is $O(nk^3 + nk^2\log n)$. 

With regards to space complexity, $\Gamma DS$ would occupy a maximum of $O(nk)$ space since each demand unit can be present in at-most $k$ heaps (used to implement $\Gamma DS$) at any stage of execution. Other than $\Gamma DS$, $NCyc\Gamma()$ and $NPath\Gamma()$ would take $O(k^2)$ space. The MinDistance heap would at most occupy $O(n)$ space. Therefore, the total space complexity of ASRAL would be $O(nk)$. 
(end of proof of Theorem \ref{asralcompl})

%{\bf Proof sketch  of Theorem \ref{dynamic}. }
%When we add a demand unit, the process is identical to one iteration of the loop in Algorithm \ref{mainalgo}. 
%When we remove a demand unit from a service center $s_j$, no new edges are added to the allotment subspace multigraph, and hence  we can skip Step 7 since no negative cycles will be introduced. Negative paths can be introduced, but these will terminate in the node corresponding to $s_j$. We can detect and remove these using a process very similar to the one in Algorithm \ref{algrnegp}, in which the edges to capture the penalty difference are leading out of $SC_\theta$. Using the argument explained in Lemma \ref{mainlemma}, we can show that the resulting allotments will be either optimal or negative cycle free. The time and space complexity arguments follow along the lines of Theorem \ref{asralcompl}. (end of proof of Theorem \ref{dynamic})

\section{LBDD instance with Schools as Service Centers}
\label{pracLBDD}
A feasible unit of measurement for the objective function in this case could be the total cost of ``operation'' for one day in Dollars. Edge costs could be the fare amount spent while traveling that edge. For penalty, we could use metrics like faculty - student ratio. If a faculty-student ratio of 1:F is maintained in a school, then its overload penalty would be Faculty-annual-salary/(F x  365) for each allotment over the capacity of the school. One can easily create more complex penalty functions by incorporating other parameters like library-books-student ratio, lab-equipment-student ratio, etc. In such cases the capacity of the school should be set to the minimum of capacities in terms total-faculty, total-equipment, total-books, etc. It should be noted that proposed algorithm is oblivious of these implementation intricacies of the penalty functions as long as these adhere to Definition \ref{pdef}.

\end{appendices}

\end{document}